\documentclass[opre,nonblindrev,twocolumn,final]{informsarxiv}

\SingleSpacedXI

\usepackage{endnotes}
\let\footnote=\endnote
\let\enotesize=\normalsize
\def\notesname{Endnotes}%
\def\makeenmark{\hbox to1.275em{\theenmark.\enskip\hss}}
\def\enoteformat{\rightskip0pt\leftskip0pt\parindent=1.275em
  \leavevmode\llap{\makeenmark}}

\usepackage{times}
\usepackage{hyperref}
\hypersetup{pdftex,colorlinks=true,allcolors=blue}
\usepackage[linesnumbered,ruled,vlined]{algorithm2e}
\usepackage{enumitem}
\usepackage{tikz,tkz-berge}
\usetikzlibrary{fit,shapes}
\newtheorem{observation}{Observation}
\newtheorem{examplee}{Example}
\newcounter{subexample}[examplee]
\renewcommand{\thesubexample}{\theexample\alph{subexample}}
\newenvironment{subexample}{
        \refstepcounter{subexample}
        \par\noindent
        \textbf{\upshape Case \thesubexample.}%
}{}
\newcommand{\Rplus}{\ensuremath{\mathbb{R}_+}}
\definecolor{myblue}{RGB}{25,25,129}
\definecolor{mygreen}{cmyk}{0.82,0.11,1,0.25}

\usepackage{natbib}
 \bibpunct[, ]{(}{)}{,}{a}{}{,}%
 \def\bibfont{\small}%
\TheoremsNumberedThrough     % Preferred (Theorem 1, Lemma 1, Theorem 2)
\ECRepeatTheorems

\EquationsNumberedThrough    % Default: (1), (2), ...
\MANUSCRIPTNO{OPRE-2019-02-070} % When the article is logged in and DOI assigned to it,
\begin{document}
%%%%%%%%%%%%%%%%

%\setcounter{page}{1}
%\DOI{10.1287/xxxx.0000.0000}%

\RUNAUTHOR{Conitzer et al.} % for four or more authors

\iffalse
\title{Multiplicative Pacing Equilibria in Auction Markets}
\fi
\RUNTITLE{Multiplicative Pacing Equilibria in Auction Markets}

% \TITLE{Bundling Information Goods of Decreasing Value}
\TITLE{Multiplicative Pacing Equilibria in Auction Markets}

\iffalse
\author[1]{Vincent Conitzer\thanks{This work was done while the author was visiting Facebook Core Data Science.}}
\author[2]{Christian Kroer\thanks{Most of this work was done while the author was at Facebook Core Data Science.}}
\author[3]{Eric Sodomka}
\author[4]{Nicolas E. Stier-Moses}
\affil[1]{%
  Econorithms, LLC and Duke University\\
  \tt{vincent.conitzer@duke.edu}
  }
\affil[2]{%
  IEOR Department, Columbia University\\
  \tt{christian.kroer@columbia.edu}
}
\affil[3]{%
  Core Data Science, Facebook Inc.\\
  \tt{sodomka@fb.com}
}
\affil[4]{%
  Core Data Science, Facebook Inc.\\
  \tt{nicostier@yahoo.com}
}
\fi

\ARTICLEAUTHORS{%
\AUTHOR{Vincent Conitzer\footnote{This work was done while the author was visiting Facebook Core Data Science.}}
\AFF{Econorithms, LLC and Duke University, \EMAIL{vincent.conitzer@duke.edu}}
\AUTHOR{Christian Kroer\footnote{Most of this work was done while the author was at Facebook Core Data Science.}}
\AFF{IEOR Department, Columbia University, \EMAIL{christian.kroer@columbia.edu}}
\AUTHOR{Eric Sodomka}
\AFF{Core Data Science, Facebook Inc., \EMAIL{sodomka@fb.com}}
\AUTHOR{Nicolas E. Stier-Moses}
\AFF{Core Data Science, Facebook Inc., \EMAIL{nicostier@yahoo.com}}
} % end of the block

\iffalse
\begin{abstract}
  \input{text/abstract}
\end{abstract}
\fi
\ABSTRACT{%
  Budgets play a significant role in real-world sequential auction markets such as those implemented by internet companies. To maximize the value provided to auction participants, spending is smoothed across auctions so budgets are used for the best opportunities. Motivated by a mechanism used in practice by several companies, this paper considers a smoothing procedure that relies on {\em pacing multipliers}: on behalf of each buyer, the auction market applies a factor between 0 and 1 that uniformly scales the bids across all auctions. Reinterpreting this process as a game between buyers, we introduce the notion of {\em pacing equilibrium}, and prove that they are always guaranteed to exist. We demonstrate through examples that a market can have multiple pacing equilibria with large variations in several natural objectives. We show that pacing equilibria refine another popular solution concept, competitive equilibria, and show further connections between the two solution concepts. Although we show that computing either a social-welfare-maximizing or a revenue-maximizing pacing equilibrium is NP-hard, we present a mixed-integer program (MIP) that can be used to find equilibria optimizing several relevant objectives. We use the MIP to provide evidence that: (1) equilibrium multiplicity occurs very rarely across several families of random instances, (2) static MIP solutions can be used to improve the outcomes achieved by a dynamic pacing algorithm with instances based on a real-world auction market, and (3) for the instances we study, buyers do not have an incentive to misreport bids or budgets provided there are enough participants in the auction.

% Enter your abstract
}%

%\KEYWORDS{deterministic inventory theory; infinite linear programming duality;
%  existence of optimal policies; semi-Markov decision process; cyclic schedule}

\HISTORY{This paper was first submitted on Feb 6, 2019 and has been with the authors for
1 year for 2 revisions.}

\maketitle
%%%%%%%%%%%%%%%%%%%%%%%%%%%%%%%%%%%%%%%%%%%%%%%%%%%%%%%%%%%%%%%%%%%%%%

% Samples of sectioning (and labeling) in OPRE
% NOTE: (1) \section and \subsection do NOT end with a period
%       (2) \subsubsection and lower need end punctuation
%       (3) capitalization is as shown (title style).
%
%\section{Introduction.}\label{intro} %%1.
%\subsection{Duality and the Classical EOQ Problem.}\label{class-EOQ} %% 1.1.
%\subsection{Outline.}\label{outline1} %% 1.2.
%\subsubsection{Cyclic Schedules for the General Deterministic SMDP.}
%  \label{cyclic-schedules} %% 1.2.1
%\section{Problem Description.}\label{problemdescription} %% 2.

% Text of your paper here

\section{Introduction}

In the last decade, auction markets have become a pervasive mechanism used by internet companies to match buyers to their target audience at the right price.
The mechanisms put in place select users matching a targeting rule that buyers specify, allowing them to bid for selected events of interest such as an impression, a click, a conversion or a video view.
This results in a winning buyer who is given the chance to show an impression and potentially generate the event of interest.
In these auction markets, buyers typically specify a budget that can be spent over a certain sequence of auctions, as well as valuations for the events of interest.
It is a responsibility of the mechanism to guarantee that the total payments of buyers do not exceed the budgets they specified.
The simplest way to take budgets into account is to bid as if there were no budget constraint, until the buyer runs out of budget. At that time, the buyer effectively stops participating in the auctions.
Unfortunately, this simple procedure is clearly not optimal: if the buyer is able to anticipate that the budget will run out well before the time period is over, it makes sense to bid less aggressively at earlier stages to be able to participate in later auctions. These later auctions, after all, may have some of the best opportunities for the buyer since, for example, they may provide the same value at a lower price.
Figure~\ref{fig:budget} shows an example
in which a buyer has a \$5 value for winning and a \$10 budget. Here, a Vickrey (second price) auction is used at each step. We assume, for simplicity, that all bids are per impression.
As shown on the left, the buyer is able to possibly win any one of the auctions for that value, but can only win the first 6~auctions before running out of budget. The buyer receives a total value of $6\times\$5=\$30$ at a cost of $\$10$, for a utility of $\$20$. Instead, as shown on the right,
the buyer can win more auctions and get a higher utility bidding $\$2$. 
The buyer wins 7 auctions for a total value of
$7\times \$5 = \$35$ at a cost of $\$10$, for a utility of $\$25$.

The previous situation motivates that auction market mechanisms more actively take budgets into account. One possibility is to perform {\em probabilistic throttling}, which consists of tossing an appropriately weighted coin for each auction. The outcome determines whether a bid is actually placed into the auction on the buyer's behalf. Selecting each probability appropriately, the buyer's budget will run out just around the end of the bidding period. Doing this for all buyers results in the process being more stable over time---as opposed to having many buyers early on and then auctions becoming thinner as buyers run out of budget, as shown on the left side of the figure.
Still, this approach also has its drawbacks. Buyers will not be considered in some auctions purely because of a coin toss, and the missed opportunities may be the ones where the buyer could have won at lower cost. Thus, this alternative may be suboptimal for buyers as well.

Another solution is to appropriately shade bids on the buyers' behalf.
(Again, for simplicity, consider a buyer who is bidding on a per-impression basis; appropriate modifications can be made for a buyer bidding on a per-click basis.)
When it appears that simply bidding the valuation $v_i$ will result in the budget being spent before the period is over, the mechanism can simply shade down each bid to $\alpha_i v_i$, where $\alpha_i \in [0,1]$ is referred to as a {\em pacing multiplier}.
An optimal multiplier will make the budget run out exactly at the end of the period, unless the buyer would not run out of budget even with $\alpha_i=1$.

\begin{figure}
\centerline{
\begin{tabular}{p{0.4\columnwidth} p{0.4\columnwidth}}
  \vspace{0pt} \includegraphics[width=1.85in]{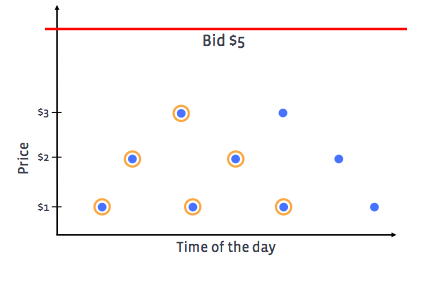} &
  \vspace{10pt} \includegraphics[width=1.85in]{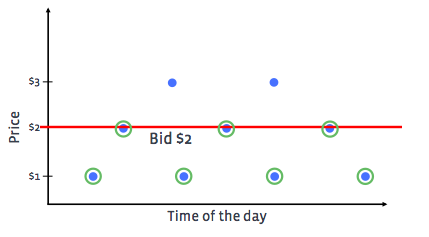}
  \\[-5mm]
\end{tabular}
}
\caption{Relation bids vs.\ total value for a budget of \$10}
\label{fig:budget}
\end{figure}

Motivated by the multiplicative mechanism which is used by several internet auction markets, we set out to study the details of the associated static game, which has not been the subject of a prior methodical study. One of the reasons that justifies its widespread use is that multiplicative pacing allows a buyer to participate in more auctions and win at lower prices, compared to probabilistic throttling.
Furthermore, \citet{balseiro2021budget} conclude that multiplicative pacing is buyer-optimal out of various options they study.

To motivate the interpretation of the mechanism as a game, note that each buyer is affected by the other buyers' multipliers. For instance, for two buyers $i$ and $j$, if buyer $i$'s multiplier $\alpha_i$ goes down, this may result in buyer $j$ winning more impressions, so $\alpha_j$ needs to go down too. Or, alternatively, it may result in buyer $j$ having to pay less for the impressions she is winning (because $j$ was setting the price for $i$, given that we use a second price auction), so that $\alpha_j$ can go up. Because the effect can work in both directions, and buyer $i$ is similarly affected by $\alpha_j$, it is not obvious that there must exist a vector of multipliers for all buyers that is mutually optimal.
The first question we address is whether there is a vector of multipliers that are {\em simultaneously} optimal for all buyers in a market with multiple {\em single-slot second price auctions}. To be optimal, all multipliers should be set so that each buyer either spends the entire budget or does not shade the bid. The choice of the vector of pacing multipliers can be viewed as the {\em equilibrium} of a {\em one-shot game} in which each $\alpha_i$ is a best response to all the other $\alpha_j$. We call such a vector of pacing multipliers, along with their corresponding allocation, a \emph{second price pacing equilibrium} (SPPE). Although in practice the multipliers are computed by the auction market on behalf of the buyers, this can still be viewed as a game since buyers can in principle change bids themselves to adjust the spending rate and even opt out of the automatic shading.

Our notion of equilibrium only stipulates that pacing multipliers should be optimal given the auction prices and winning bids, and it is thus not the case that our equilibria are equivalent to Nash equilibria in the one-shot game. In order to specify what a Nash equilibrium would be we would additionally need to specify what happens under deviations: a single buyer changing their pacing multiplier could change prices in auctions they are not winning, thus causing other buyers to exceed their budget. If the game is such that this budget exhaustion does not cause the budget-exhausted buyer to be dropped from some auctions (and thus this does not cause reduced prices) then our equilibria constitute Nash equilibria, but if they are dropped from some auctions then there may be an incentive to cause such dropping in order to reduce prices.

We prove that an equilibrium always exists (which does not follow
from existing results due to discontinuities when there are ties or when budgets are exceeded) and
that a pacing game can admit multiple equilibria that are not outcome
equivalent, which leads to equilibrium selection issues. We compute equilibria
with respect to commonly-studied objective functions such as social welfare and
revenue to provide insights on the gaps between best and worst equilibria. 
We show examples where this gap can be quite large. Then, we study the complexity of finding equilibria, and provide a {\em mixed-integer program} (MIP) to find them.
Using the MIP we study the equilibrium-selection issues empirically, and find that equilibrium multiplicity is rarely an issue across both randomly-generated and real-world instances.
We complement the MIP with best-response and regret-based dynamics as alternative computational tools for finding equilibria.

We go on to show that pacing
equilibria are a refinement of {\em competitive equilibria}. A competitive
equilibrium consists of good prices and allocations such that each buyer
obtains a bundle that she considers optimal given those prices, and all goods
with positive prices are completely allocated. We show that every pacing
equilibrium is also a competitive equilibrium. Moreover, for every competitive
equilibrium, it is possible to add some non-winning buyers so that
it becomes a pacing equilibrium. We exhibit an example in
which the unique pacing equilibrium is not revenue-minimizing among competitive
equilibria, i.e., there is another competitive equilibrium with lower revenue.
This, in combination with the previous result, implies revenue-nonmonotonicity
in the bids, i.e., additional bids can reduce the revenue of pacing equilibria.

The (near) equivalence between our pacing equilibrium and competitive equilibrium leads to another motivation of our work: since real auction markets happen dynamically over time, it is not necessarily clear what the resulting allocations and prices will look like.
However, ex-post one may hope that the allocations and prices roughly constitute a competitive equilibrium.
This would imply several important properties such as envy-freeness, Pareto efficiency, and that market-clearing prices were used.
Our paper shows that such a competitive equilibrium is achievable by having the seller conduct second price auctions, and letting the buyers (or proxy bidders) use multiplicative pacing.
This lends support to the approach often taken in real auction markets, where a proxy bidder attempts to identify the correct pacing multiplier over time via some adaptive control algorithm.

Since there are many unknowns in real-world auction markets (e.g., auction
participants, user visits, resulting prices, event realizations, etc.),
practical mechanisms learn the optimal multipliers by dynamically adjusting
them using forecasts of when the budget will run out. In
our theoretical model, we sidestep the issue of dynamically adjusting the
multipliers, and consider the limit case in which the auctioneer can perfectly
predict the impressions that will arrive.
Although the one-shot game assumes away the stochastic and dynamic elements,
the results we obtain for this limit case have clear implications for real-world auction markets.
To address that, we investigate an adaptive pacing setting, and show
that the regret-based adaptive pacing algorithm of \citet{balseiro2017dynamic}
finds an allocation that is close to the solution of our MIP.
Using realistic instances inspired by auctions on the internet, we find that
the outcome in the adaptive setting can be improved by seeding the adaptive
dynamics with the MIP solution, even though the MIP solves a static instantiation of
the time-varying auctions.

To create realistic instances for computational studies, we sample impressions from real auctions and generate a bipartite graph that encodes their structure. Subsequently, we cluster the graph to reduce its size
without losing the important competitive information that describes the auction market.
The procedure to create small instances that capture the intricacies of the market and seeding dynamic mechanisms with the resulting equilibria may pave the road to practical use of pacing equilibria in real-world markets, in addition to learning optimal multipliers using dynamics.
This observation motivated \citet{kroerApproxMarketEq} to study, in follow-up work, how to solve simpler representations of dense instances of competitive equilibrium problems, and how solutions to an approximation differ from the original ones.

We employ the MIP solution procedure to study incentive compatibility properties of the pacing mechanism studied here.
Generating ground-truth values and budgets for buyers, we compute pacing equilibria when they misrepresent their types.
Our study provides evidence that incentives to misreport bids and budgets are weak, provided that there are enough participants in the auctions. 

Finally, let us add a remark on the motivation of the paper. 
A number of real-world platforms started to use multiplicative pacing because they realized that individual pacing multipliers have attractive properties from the perspective of an individual buyer. 
The motivation was that by lowering bids in this way, utilities would generally be higher, since the benefits accrued when obtaining impressions, clicks or conversions is held constant. This, in turn, led to platforms implementing control mechanisms to learn those multipliers. 
Nevertheless, there was not any principled theory showing that this would lead to solutions that are, in aggregate, of high quality, or that this would lead to equilibrium points. (Even if, from an individual buyer perspective, pacing multipliers can be motivated via Lagrangian duality on the budget constraints. Note that the multiplicative pacing approach can be given a Lagrangian interpretation as follows. Consider the problem of optimal bidding in hindsight for an individual buyer with a budget constraint. A Lagrangian relaxation of the budget constraint leads to a simple optimization problem where the buyer wants to buy all goods that have positive valuation after accounting for the Lagrangian ``price.'' This solution is obtained by setting $\alpha_i = \frac{1}{1+\mu}$ and bidding $\alpha_i v_{ij}$ for every good $j$, where $\mu$ is the Lagrange multiplier on the budget constraint.) 

Our motivation when we started working on the present paper was to provide a framework, leaning on game-theoretic principles, to explain and analyze what the platforms had already done. This complements seminal research that already existed, but also intersects with research that was being done concurrently, as discussed in the literature review.
Our results contribute evidence that multiplicative pacing is an appropriate mechanism to manage budgets. Equilibrium multipliers are guaranteed to exist and the MIP we propose can be used to guide equilibrium selection so buyers can jointly maximize their utility by bidding consistently within their budgets. In addition, according to our computational study, the mechanism is approximately incentive compatible when auctions have enough participants. 

\paragraph{Presentation of results}

We start by framing our model with respect to the existing literature in Section~\ref{sec:related}.
Then, we introduce the pacing game and define our equilibrium concept in Section~\ref{sec:pacinggame}.
Section~\ref{sec:eqmanalysis} discusses the details of these equilibria, including existence, sensitivity and multiplicity, followed by a connection to competitive equilibria in Section~\ref{sec:competitive_equilibrium}.
Section~\ref{sec:hardness} presents results related to computability of equilibria including computational complexity, iterated best responses, and a MIP formulation.
Finally, we provide an empirical illustration through computational experiments in Section~\ref{sec:experiments}.
After describing the instances we consider, we study the scalability of our MIP formulation, we evaluate uniqueness empirically, we explore how robust pacing equilibria are to misreporting true values and budgets, and finally we put the pacing equilibrium concept in perspective by evaluating it through a dynamic algorithm that incorporates time into the model.
We present some final thoughts in Section~\ref{sec:conclusion}.
We refer the reader to the e-companion of this article which includes missing proofs, additional discussion, model tweaks, and further examples and experiments.

\section{Related work}\label{sec:related}

There is a large literature on casting the delivery of online advertising under budget constraints as a centralized matching problem that assigns advertisers to impression opportunities, rather than taking the perspective of auctions and strategic behavior as we do.
\citet{mehta2007adwords} considers an online setting and introduce an algorithm that has a competitive ratio of $1-1/e$ for revenue maximization when the volume and sequence of queries is unknown.
\citet{Abrams2007optimal} investigate a linear programming approach based on column generation, where each column is a slate of ads that can be considered for an impression opportunity, and optimizes efficiency or revenue while controlling advertiser spend within the time horizon. Considering slates allows the model to price according to GSP.
Additional papers that generalize and extend these results include \citet{feldman2010online,devanur2011near,bhalgat2012online}.
More recently, \citet{asadpour-concisebid} considers the case in which there are several budget constraints that apply to different subsets of ads that an advertiser is running. Since managing bids at the auction level is difficult, they propose a system with concise bidding strategies that splits opportunities in clusters and bids uniformly for each of them. They develop a constant-factor approximation algorithm to optimize the strategies for a fixed number of clusters.

Several articles also consider a stochastic version of the matching problem \citep{goel2008online,devanur2009adwords,feldman2009online,feldman2010online,devanur2011near,mahdian2012online,devanur2012asymptotically,mirrokni2012simultaneous}, including embedding the matching in a game
\citep{charles2013budget}.
The approaches in these articles match supply and demand directly, rather than having all candidate ads compete to determine the winner through an auction, thus they are not directly applicable to our setting.

Another line of research considers how individual buyers should optimize their
budget spending across a set of auctions. This has
been cast as a form of knapsack problem
\citep{Feldman2007budget,borgs2007dynamics,chakrabarty2008budget}, a Markov
Decision Process~\citep{amin2012budget,gummadi2013optimal}, constrained
optimization~\citep{zhang2012joint,zhang2014optimal}, and optimal control~\citep{xu2015smart}.
\citet{gummadi2013optimal} consider a Markov Decision Process formulation of the budget optimization problem.
From the perspective of an advertiser with a budget constraint competing with a set of iid ({\em independently and identically distributed\/}) bids in a second price auction or GSP setting, the optimal policy is to multiplicatively shade the value of the impression.
\citet{agarwal2014budget} describe a practical implementation with experiments on LinkedIn advertising data.
\citet{ciocan-endogenousbudgets} take a fresh perspective by considering endogenous budget decisions arising from the cost of capital for advertisers, and allowing them to strategize the selection of ad bids and campaign budget. The platform then selects a winner and a runner-up that sets the price for each fractional allocation using a linear program.

The closest paper to ours is a groundbreaking paper by \citet{balseiro2021budget}, which was written independently. They define equilibria for a variety of budget management procedures, including multiplicative pacing, and prove the existence of equilibria. This is related to the existence result we provide; the main difference is that they assume independent and continuous valuation distributions and as a result they effectively assume away ties.
In contrast, we need to pay special attention to how ties are broken; specifically, how much of each good goes to each tied buyer.
These fractions are a fundamental part of what constitutes an equilibrium in our setting. 
(See the model's description in the next section for a discussion on how to interpret fractions.)
Ties in the bids are not a measure-zero event in our setting, because pacing parameters will often result in ties even for generic valuations. 
\citet{balseiro2021budget} introduce an iterative algorithm based on the buyers repeatedly best-responding that is not always guaranteed to converge to equilibrium and evaluate it in experiments. We show that in our setting such an algorithm can cycle, give an exact MIP formulation for finding optimal equilibria (and show that these problems are NP-hard), and evaluate it in experiments.

\citet{balseiro2017dynamic} study how an individual buyer might adapt their
pacing multiplier over time. They study a stochastic setting, where each buyer
has valuations drawn at each time step independently of time and the other
buyers (though they show that they can also support imperfect correlation
between buyers under certain technical conditions). They design
regret-minimizing algorithms for their setting, and show asymptotic optimality
under adversarial and stationary settings. Their setting is different from ours
in that it is dynamic, it requires independence of valuations, and it requires
the distribution of valuations to be absolutely continuous. For these reasons
their algorithm is not guaranteed to work in an adaptive variant of our setting.
Nonetheless, we show in our experimental setting that their algorithm can
achieve strong performance when combined with good initial pacing multipliers
from solutions to our MIP model.

\citet{balseiro2015repeated} investigate budget-management in auctions through a \emph{fluid mean-field} approximation, which leads to elegant existence results and closed-form descriptions of equilibria in certain settings. Again, this differs from our setting in that they effectively assume away ties by making distributional assumptions on the payments faced by the buyers.
That paper and \citet{balseiro2021budget} 
also assume that for a given impression, the valuation of each buyer is independent from that of other buyers. We require no such assumption.

Rather than trying to adapt variants of second price auctions through budget
smoothing, one can design entirely new mechanisms that handle budgets
directly~\citep{ashlagi2010position,bhattacharya2010incentive,dobzinski2012multi,goel2015polyhedral,goel2015clinching}.
However, for practical purposes we focus on methods that implement
second price auctions, as these tend to be preferred in real-world auction markets.

Finally, the relationship between auctions and competitive equilibria has
been explored in some other contexts. \citet{klemperer2010product} uses
competitive equilibrium as the allocation mechanism in \emph{product-mix
auctions}. Conversely, auction-based algorithms have been used for arriving at
competitive equilibrium in certain
contexts~\citep{garg2004auction,garg2006auction,kapoor2007auction,nesterov2018computation}.
In a follow-up to the present work, \citet{conitzer2018pacing} show that
\emph{first price pacing equilibria} can also be interpreted as competitive
equilibria, and in particular they correspond to solutions to
the Eisenberg-Gale convex program in the quasi-linear
case~\citep{eisenberg1959consensus,cole2017convex}.

\section{Pacing Games for Auction Markets}\label{sec:pacinggame}

In this section we define the pacing games that will be the focus of our work.
We consider a market in which a set of buyers $N=\{1,\ldots,n\}$ target a set of divisible goods $M=\{1,\ldots,m\}$. 
An instance of the game is defined by a set of valuations and budgets: 
each buyer $i$ has a valuation $v_{ij}\ge 0$ for each good $j$, and a budget $B_i>0$ that constrains the spend across all goods.
Figure~\ref{fig:example_games} shows two examples of pacing games, which will be solved later.

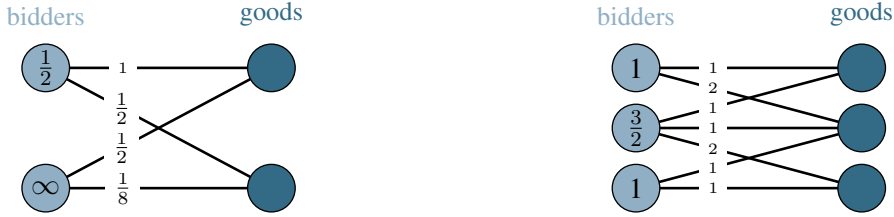
\begin{figure}
  \small
\definecolor{blush}{RGB}{208,150,131}
\definecolor{taupe}{RGB}{115,96,91}
\definecolor{mist}{RGB}{144,175,197}
\definecolor{stone}{RGB}{51,107,135}
  \begin{minipage}{.47\columnwidth}
    \centering
    \begin{tikzpicture}[]
      \GraphInit[vstyle=Normal]
      \SetVertexNoLabel
      \SetUpVertex[Math,Lpos=-180,LabelOut]
      \SetVertexNormal[FillColor=mist,OuterSep=0pt,TextColor=myblue]
      \grEmptyPath[form=2,x=0,y=0,RA=1.6,rotation=90,prefix=B]{2}
      \AssignVertexLabel{B}{$\infty$,$\frac{1}{2}$}
      \SetVertexShade[BallColor=stone,OuterSep=0pt]
      \SetUpVertex[Lpos=0]
      \SetVertexNormal[FillColor=stone,TextColor=mygreen]
      \SetVertexNoLabel
      \grEmptyPath[form=2,x=3,y=0,RA=1.6,rotation=90,prefix=G]{2}
      \SetUpEdge[lw=1pt,color=black]
      \Edge[label = $\frac{1}{8}$,style={pos=.30}](B0)(G0)
      \Edge[label = $\frac{1}{2}$,style={pos=.30}](B0)(G1)
      \Edge[label = $\frac{1}{2}$,style={pos=.30}](B1)(G0)
      \Edge[label = \tiny$1$,style={pos=.30}](B1)(G1)
      % the set U
      \node [mist,fit=(B0) (B1),label=above:\textcolor{mist}{\small buyers}] {};
      % the set V
      \node [stone,fit=(G0) (G1),label=above:\textcolor{stone}{\small goods}] {};
    \end{tikzpicture}
  \end{minipage}
  \begin{minipage}{.47\columnwidth}
    \centering
    \begin{tikzpicture}[]
      \GraphInit[vstyle=Normal]
      \SetVertexNoLabel
      \SetUpVertex[Math,Lpos=-180,LabelOut]
      \SetVertexNormal[FillColor=mist,OuterSep=0pt,TextColor=myblue]
      \grEmptyPath[form=2,x=0,y=0,RA=0.8,rotation=90,prefix=B]{3}
      \AssignVertexLabel{B}{1,$\frac{3}{2}$,1}
      \SetVertexShade[BallColor=stone,OuterSep=0pt]
      \SetUpVertex[Lpos=0]
      \SetVertexNormal[FillColor=stone,TextColor=mygreen]
      \SetVertexNoLabel
      \grEmptyPath[form=2,x=3,y=0,RA=0.8,rotation=90,prefix=G]{3}
      \SetUpEdge[lw=1pt,color=black]
      \Edge[label = \tiny$1$,style={pos=.30}](B2)(G2)
      \Edge[label = \tiny$2$,style={pos=.30}](B2)(G1)
      \Edge[label = \tiny$1$,style={pos=.30}](B1)(G2)
      \Edge[label = \tiny$1$,style={pos=.30}](B1)(G1)
      \Edge[label = \tiny$2$,style={pos=.30}](B1)(G0)
      \Edge[label = \tiny$1$,style={pos=.30}](B0)(G1)
      \Edge[label = \tiny$1$,style={pos=.30}](B0)(G0)
      % the set U
      \node [mist,fit=(B0) (B2),label=above:\textcolor{mist}{\small buyers}] {};
      % the set V
      \node [stone,fit=(G0) (G2),label=above:\textcolor{stone}{\small goods}] {};
    \end{tikzpicture} 
  \end{minipage}
\caption{Two examples of pacing games. Buyers and goods are represented by vertices in a bipartite graph on the left and right, respectively.
  The labels on buyer vertices represent budgets, while the labels on edges denote the buyers' valuation for the good (missing edges denote null valuations).
}
\label{fig:example_games}
\end{figure}

A platform enables the transactions through which the goods are sold. 
To allocate winners and set prices, it relies on a mechanism based on independent (single slot) second price auctions, one for each good.
Our goal is to study how buyers can smooth out their spending across all of the auctions to make their bids and budgets compatible with each other.
This is achieved via \emph{multiplicative pacing}: each buyer selects a \emph{pacing multiplier} $\alpha_i \in [0,1]$, which they use to compute the effective bid of $\alpha_i v_{ij}$ for each good $j$.
We refer to the resulting bids as {\em multiplicatively paced}.
Thus, the strategy space of buyer $i$ in the pacing game is the set of possible pacing multipliers $\alpha_i \in [0,1]$. 

Given a vector of pacing multipliers $\alpha \in [0,1]^n$ chosen by all buyers, the platform allocates goods and prices according to second price auctions with the bids mentioned earlier.
Thus the allocation rule is mostly straightforward: allocate to the highest paced bid. However, this still leaves open the question of how to deal with tied bids. 
This is relevant because multipliers allow buyers to manipulate the effective bids, which is likely to create ties in equilibrium. 

For the majority of the results, our solution concept will consist of a vector of pacing multipliers together with an allocation of goods to buyers that breaks ties in suitable fashion. 
However, for the purpose of explaining how our solution may be viewed as a game, we need to make the tie-breaking rule part of the rules of the game, so that the strategy space of each buyer consists only of choosing a pacing multiplier.
To that end, we assume that there is some fixed allocation rule $x_{ij}(\alpha)$ which denotes how much buyer $i$ receives of good $j$, given a vector of pacing multipliers $\alpha \in [0,1]^n$.
The resulting price for good $j$ arises from its corresponding second price auction, considering the paced bids:
$p_j =\max_{k\neq\text{max bidder on $j$}} \alpha_k v_{kj}$.
The allocation rule must satisfy the following conditions:
  \begin{itemize}
  \item $x_{ij}(\alpha) > 0 \Rightarrow \alpha_iv_{ij} = \max_{k} \alpha_{k}v_{kj}$ for all $i,j$
  \item $\sum_i x_{ij}(\alpha) = 1$ for all $j$
  \item If there exists an allocation $x'$ such that $\sum_j x'_{ij}p_j \leq B_i$ for all $i$, then $\sum_j x_{ij}(\alpha)p_j \leq B_i$ for all $i$
  \end{itemize}
The first and second conditions enforce consistency with second price auctions, whereas the last condition says that if it is possible to break ties such that all budgets are satisfied, then $x(\alpha)$ should do that.

The utility for buyer $i$ resulting from pacing vector $\alpha$ is a regular quasilinear utility function, except that their budget constraint must be satisfied:
  \[
    u_i(\alpha) =
    \begin{cases}
      \sum_j (v_{ij} - p_j) x_{ij}(\alpha) & \text{ if } \sum_j p_jx_{ij}(\alpha) \leq B_i \\
      -\infty & \text{ otherwise}
    \end{cases}.
  \]

The utility of the platform is inconsequential to the equilibrium analysis, but of course useful when comparing solutions with each other. In Section~\ref{sec:eqmanalysis}, we will get back to this, considering revenue and welfare. We end the description of the \emph{pacing game} with the sequence of events that defines it:
\begin{enumerate}
  \item Instance created; each buyer learns their own valuations and budget.
  \item Each buyer selects an optimal pacing multiplier anticipating that other buyers do so as well. In order to do this, we assume that each buyer knows the valuations and budgets of other buyers.
  \item The platform gets bids from buyers, runs the mechanism, and produces prices and an allocation. Ties are broken so as to preserve budgets, if possible.
  \item Buyers get allocations, realize spends, and perceive utilities.
\end{enumerate}

Having a game with full information can be justified in two ways.
First, auctions happen dynamically and buyers have the chance to revise their bids over time. As the interactions occur, buyers can adapt their pacing multiplier until they reach the right rate of expenditure. In order for this to occur it is enough to observe prices; details about competitors are not needed. This is a common justification for (static) Nash equilibria; here the notion of a repeated game is natural since the number of auctions is typically large.
In addition, a proxy bidder that acts on behalf of the buyer might have additional context of what other buyers may do. At the very least, even if the proxy bidder does not have access to any information about other buyers, the dynamic interpretation applies since it is reasonable that the proxy bidder continuously monitors the performance of the auctions in which that buyer participates, and adjusts the pacing multiplier accordingly. Below we provide more details about the platform acting as a proxy bidder.

We assumed that ties can be broken to satisfy the budgets of buyers. Let us see why tie-breaking is important via the following example.

\begin{example}[Ties may make buyers overspend budgets under arbitrary allocations]

As depicted on the left side of Figure~\ref{fig:example_games},
assume that $v_{11}=1$, $v_{12}=1/2$, and $B_1=1/2$, while
$v_{21}=1/2$, $v_{22}=1/8$, and $B_2=\infty$.
If independent auctions with independent bids were used, then buyer 1 could guarantee a utility of $1/2$ by bidding $1$ on good $1$ and $0$ on good $2$ and stay within budget.

Assume that buyer $1$ wins (some of) good $1$ with a multiplicatively paced bid. This implies that $\alpha_1 \geq 1/2$, from which it follows that $\alpha_1 v_{12} \geq 1/4 > 1/8$, thus making her win good $2$ as well. 
But without enforcement of budgets constraints, under second price auctions, the buyer would pay $1/2+1/8 > B_1$, resulting in a budget violation.
In that case, to be safe, she should set $\alpha_i < 1/2$ and lose good $1$, resulting in a utility of at most $1/2-1/8=3/8$. 

If the buyer were sure that budgets are always satisfied if possible, she would bid $\alpha_1=1/2$, the tie would be broken to win $3/4$ of good $1$ and all of good $2$, for a combined valuation of $3/4+1/2=5/4$ and a combined payment of $3/8+1/8 = 1/2 = B_1$.
This results in a utility of $3/4$, which is the best possible.
\end{example}

As an implementation remark for our stylized model, a tie-breaking rule that guarantees budgets to be satisfied is natural in real-world internet ad auction markets.
One way to map the model to reality is to consider that goods of an instance are split into many small units. A buyer could control the fractions she wins when tied as winner by slightly modifying the bids on these individual units as required.
This perspective can be thought of as mapping goods to impression types, and those subdivisions to particular impressions.
These small bid modifications that result in budget-satisfying allocation are likely to happen automatically in those real auction markets. The reason is that the pacing multipliers are estimated dynamically and in real time. The fluctuations of multipliers over time make realized bids fluctuate as well and affect allocations.
Alternatively, tie breaking can be handled in practice by perturbing each bid with independent noise before computing the allocation of each good. If only the expected value and expenditure of an allocation are considered, then this can provide a solution concept very similar to ours.

To start discussing equilibrium concepts in this game, which will lead to pacing equilibria, it is important to provide evidence that a single dimensional search space for each buyer is appropriate, since that is a crucial justification of pacing multipliers and the definition of the game.
The following proposition shows that relying on multiplicative pacing is in the best interest of buyers.
In other words, the set of best responses in the more general space of strategically modifying all valuations always intersects with the multiplicatively-paced bid vectors. This fact requires ties to be broken in favor of the buyer in question, which contradicts our definition of a fixed $x(\alpha)$ function. We will eventually define our pacing equilibria in a way that makes this a non-problem.

\begin{proposition}\label{prop:pacing_optimality}
Suppose we allow arbitrary bids in each auction, i.e., the bids $b_{ij}$ are not necessarily multiplicatively paced.
Then, holding the bids of all other buyers in all auctions fixed, each buyer $i$ has a best response that is multiplicatively paced
(assuming that, when she is tied to win a good, she can choose the fraction of the good she wins).
\end{proposition}
\begin{proof}{Proof.}
Consider a best response by buyer $i$ consisting of bids $b_{i1}, \ldots, b_{im}$.  Let $\alpha_i^{\max} = \max_{j} b_{ij} / v_{ij}$, and without loss of generality suppose $\alpha_i^{\max}$ is minimized among best responses for buyer $i$.  We will show that bidding $b'_{ij}=\alpha_i^{\max}v_{ij}$ is also a best response.  Suppose not.  Clearly $\alpha_i^{\max} \leq 1$ since it never helps to bid more than one's valuation.  Hence $b'_{ij} \leq v_{ij}$ for all $j$.  Because we have $b'_{ij} \geq b_{ij}$ for all $j$, then $i$ can only be winning more goods, at prices below her valuations.  Hence the only way in which the $b'_{ij}$ can fail to be a better response than the $b_{ij}$ is by exceeding $i$'s budget. 
Because by assumption $i$ can break ties as she wishes, it follows that with the bid $b'_{ij}$ she exceeds her budget even if she accepts none of the goods for which she is tied.
Because the bid $b_{ij}$ did not cause the allocation to exceed the budget, it follows that there exists a good $j^*$ with price (highest other bid) $p_{j^*}$ such that $b_{ij^*} \leq p_{j^*} < b'_{ij^*}$ of which $i$ was not winning everything when bidding $b_{ij}$.
Now consider gradually increasing $b_{ij^*}$ towards $b'_{ij^*}$ (or increasing the fraction of $j^*$ that $i$ is allocated). 
If under bid $b_{ij}$ the budget was not already exhausted, then the moment
that $i$ starts winning some of $j^*$ (at a price below her valuation), we have found a better response and hence the required contradiction.
If the bid $b_{ij}$ did already cause the buyer to exhaust the budget, then once $i$ starts winning some of $j^*$, 
we can pay for this by reducing the amount spent on some good $j^{**}$ with $p_{j^{**}}=\alpha_i^{\max}v_{ij^{**}}=b_{ij^{**}}$.  (Such a good must exist by the minimality of $\alpha_i^{\max}$: if $x_{ij} = 0$ for all $j$ such that $\alpha_i^{\max}v_{ij^{**}}=b_{ij^{**}}$ then a different best response with all those bids set to zero exists, contradicting the minimality of $\alpha_i^{\max}$).
The utility buyer $i$ receives per dollar spent on $j$ is $(v_{ij}-p_j)/p_j = v_{ij}/p_j - 1$.
But we have $p_{j^{**}}/v_{ij^{**}} = \alpha_i^{\max}$ and $p_{j^*}/v_{ij^*} < \alpha_i^{\max}$.
Hence $v_{ij^{**}}/p_{j^{**}} - 1  = 1/\alpha_i^{\max} - 1 < v_{ij^*}/p_{j^*}$, i.e.,
the bang-per-buck is actually higher on $j^*$.  So shifting spending to $j^*$ is utility-improving, giving us the required contradiction.
\hfill\Halmos
\end{proof}

A \emph{pure-strategy Nash equilibrium} (PNE) is a vector of pacing multipliers $\alpha$ such that every buyer is best responding, i.e., $u_i(\alpha) \geq u_i(\alpha_i', \alpha_{-i})$ for all $\alpha_i' \in [0,1]$. More concretely, in a PNE $\alpha$ each buyer must be within budget, and if they are not spending their entire budget then every $\alpha_i' > \alpha_i$ must be such that they either win the same goods as under $\alpha_i$, or break their budget.

Unfortunately, we cannot follow the standard approach and adopt PNE as our equilibrium concept since some of them exhibit undesirable properties, as the next example illustrates.
\begin{example}
Consider one good and two buyers. Buyer 1 has valuation 10 and buyer 2 has valuation 5. Each buyer has budget 5. A subset of the PNE are all pairs $(\alpha_1,\alpha_2)$ such that $\alpha_1=1$, since buyer 1 wins the whole good and pays $5\alpha_2$, whereas buyer 2 cannot improve their utility no matter what pacing multiplier they choose.
\end{example}
In this example budgets play no actual role: neither player needs to smooth their spending, and a regular second price auction yields a good outcome. In order to avoid these undesirable PNE solutions, we introduce the following principle:
\begin{definition}[No unnecessary pacing]
The no unnecessary pacing principle states that if a buyer does not spend their whole budget, then there should not be pacing (i.e., the pacing multiplier should equal one).
\end{definition}

Intuitively, the no unnecessary pacing principle makes sense from the perspective of an individual buyer: if their budget is not binding, then they should act as a normal buyer in a second price auction, where truthful bidding is a dominant strategy.

We will restrict our attention to PNEs that satisfy the no unnecessary pacing principle, leading to a \emph{refinement} of the set of all PNEs of the pacing game. A priori, it is unclear whether there will always exist an equilibrium under this refinement (it is not evident that even a PNE must exist). However, as we shall see, an equilibrium under this refinement is guaranteed to exist as long as ties are broken such that no unnecessary pacing holds when possible. It turns out that refined equilibria can be characterized purely in terms of budget feasibility and the no unnecessary pacing principle since these two conditions imply that buyers must be best responding to each other (see Proposition~\ref{the:best_responding}).

Motivated by the refinement of PNE, we now define our notion of equilibrium formally, relying on the conditions it must satisfy. Then, we prove that this notion coincides with the refinement of PNE that we discussed earlier.
In a standard game-theoretic setting, we would simply rely on our refinement given by any PNE satisfying no unnecessary pacing. However, the truly game-theoretic pacing game setup requires us to specify the tie-breaking rule for every pacing vector, which is not practical. Secondly, Proposition~\ref{prop:pacing_optimality} would then fail to hold since it requires buyers to be able to choose their tie-breaking allocation.
Instead, we approach this problem as a competitive market and explicitly make the allocations part of the equilibrium definition.

This makes our pacing equilibrium notion similar to that of competitive equilibria, while still retaining best-response properties. (We will expand on the connections to competitive equilibria in Section~\ref{sec:competitive_equilibrium}.) From the perspective of the pacing games we have discussed so far, our pacing equilibrium definition captures the PNEs satisfying no unnecessary pacing for \emph{every} pacing game where the tie-breaking rule agrees with the allocation used in the pacing equilibrium.

\begin{definition}[Pacing equilibrium]\label{def:pacing}
  A second price pacing equilibrium (SPPE) consists of a vector of pacing multipliers $\alpha \in [0,1]^N$, and
  fractions $x_{ij}\in[0,1]$ indicating allocations of good $j$ to buyer $i$.
  These elements need to satisfy budget constraints, no unnecessary pacing, the feasibility of the allocation and that prices emerge from a second price auction, all expressed in the following conditions:
\begin{itemize}
\item For all $j$, $\sum_i x_{ij} \leq 1$ (with equality if there is at least
  one $i$ with $v_{ij}>0$); also, for all $i$ and $j$, $x_{ij}>0$ implies that
  $i$'s bid $\alpha_{i}v_{ij}$ was (possibly tied for) the highest on $j$.
\item If $x_{ij}>0$, then the per-unit price $p_{j}$ is the highest bid $\alpha_{i'}v_{i'j}$ other than $i$'s bid.
\item For all $i$, $\sum_j s_{ij} \leq B_i$,
  where $s_{ij}=p_{j}x_{ij}$ is the total spend of buyer $i$ in good $j$.
  In addition, if the inequality is strict, then $\alpha_i=1$.
\end{itemize}
\end{definition}

Since all our results focus on the second price auction, we will sometimes refer to an SPPE simply as a {\em pacing equilibrium}, with the understanding that it is based on second price auctions. This is in contrast to results by \cite{conitzer2018pacing} that look at pacing equilibria for both first and second price auctions.

Note that this definition of pacing equilibrium does not explicitly require that buyers are best responding. 
Nonetheless, the conditions ensure that each buyer is best responding, thus justifying the solution concept.

\begin{proposition}
  Consider a pacing equilibrium $\{ \alpha_i, x_{ij}\}_{i\in N, j\in M}$.
  For each buyer $i\in N$, the pacing multiplier $\alpha_i$ is a best response to the paced bids of all other buyers (even if, when choosing their best response multiplier, they can choose how ties are broken).
\label{the:best_responding}
\end{proposition}
\begin{proof}{Proof.}
  Consider an arbitrary buyer $i\in N$. We will consider two cases. When
  $\alpha_i=1$, bids equal values for all goods. By the
  properties of the second price auction, this buyer cannot gain additional
  utility by raising or lowering their bid. When $\alpha_i < 1$, the
  buyer is guaranteed to be spending their entire budget by the definition of a
  pacing equilibrium.
  Raising $\alpha_i$ causes overspending if additional goods are won, which
  yields $-\infty$ utility. If no additional goods are won, then it
  has no effect on the utility of buyer $i$. Conversely, if buyer $i$ lowers
  $\alpha_i$, the only thing that can happen is winning fewer goods. Since the
  buyer is already bidding less than their true valuation, this can only cause
  them to lose goods that they gained positive utility from winning.
\hfill\Halmos
\end{proof}

The previous result shows that pacing equilibria yield PNEs with no unnecessary pacing for any pacing game with a tie-breaking rule that yields the allocation in the pacing equilibrium. To see why any PNE with no unnecessary pacing yields a pacing equilibrium, note that the first two definitions of pacing equilibrium are satisfied by the design of a pacing game, while the third condition is simply the ``no unnecessary pacing'' refinement condition.

\subsection*{Discussion on the Model Setup}

In the definition of the game we took the perspective that it is the buyers who choose pacing multipliers, but as mentioned in the introduction, we are primarily motivated by the setting where the platform performs the budget management on behalf of each buyer, via \emph{proxy bidders}.
The proxy-bidding system is built with the intent that the average buyer does not waste effort in designing complex bidding strategies, so they can focus their energy in improving the value they provide to their users. 
In a pure proxy-bidder setting, the platform is assumed to have access to the (true) budgets and (true) valuations $v_{ij}$, and its goal is simply to implement budget-smoothing via multiplicative pacing. Since we assume that the game is full information, the proxy bidders act on behalf of the buyers, compute a pacing equilibrium, and submit the same bid as the buyer would have submitted. Item (2) in the game play explained earlier is subdivided as follows: 

\begin{enumerate}
  \item[2'] Each buyer submits their true valuations and budgets to the proxy bidder operated by the platform, if they so desire.
  \item[2''] The proxy bidder computes an optimal pacing multiplier on behalf of the buyer.
\end{enumerate}

There are a few issues that are important to discuss about proxy bidders.
First, the platform gives buyers the option to pace but it is their choice to do so. The platform anticipates that if it does not do what is best for the buyer, they will not opt in. Notice that although buyers may not have access to auction-by-auction outcomes, they can easily experiment with simultaneous campaigns to learn the optimal bidding strategy. 
This is why the platform seeks an equilibrium by optimizing buyers' utilities, as opposed to attempting to find a centralized solution that maximizes welfare or revenue. 
As the feature is opt-in, some buyers select not to use pacing because of various reasons, and in practice there is a mix of buyers pacing themselves and buyers using proxy bidders.

Second, we assume that the game is full information and do not consider strategic issues of buyers misreporting their budgets or valuations. 
A partial justification for this comes from the best-response properties of pacing equilibrium. Nonetheless, it is still possible that buyers may shift the pacing equilibrium computed by the proxy bidders by misreporting.
In Section~\ref{sec:misreporting}, we investigate the extent to which buyers can gain utility from this type of manipulation. We find that market thickness quickly makes it impossible for buyers to significantly improve their utility by misreporting.

Considering a practical implementation of the above, in the internet ad markets typically buyers are not able to submit their valuation vectors; they do not even know them exactly. Instead, they would submit their \emph{value-per-click} $v_i$, budget $B_i$, and targeting criteria specifying which user segments they are interested in. The valuation for an impression $j$ that fits the targeting criteria would be calculated as $v_{ij} = \gamma_{ij} v_i$, where $\gamma_{ij}$ is the \emph{click-through rate} of impression $j$ (i.e. the probability that the user clicks on the ad). The click-through rate would typically be estimated by the platform, and thus not modifiable by the buyer. To this point, and as mentioned in the previous paragraph, Section~\ref{sec:misreporting} investigates whether buyers have incentive to misreport, both in the setting where they report their value-per-click, as well as the setting where they report their entire valuation vector.

Finally, we highlight that our model is fully static, so values, prices and everything else are all realized at the same time. Our solution concept gives us an ex-post notion of what we would like the market outcome to be, similar to the setting of a competitive equilibrium.
In terms of predictions needed to feed the model, a platform is in a good position to use it to evaluate a market. Campaigns are submitted beforehand, budgets refresh daily, and impression opportunities arise from users logging onto the platform, which can be predicted with accuracy.
In the computational section, we introduce a model with dynamics as a more realistic version that can be used for evaluation. 
We provide evidence that the dynamics approach the (static) pacing equilibria, for various initial conditions.
Moreover, if we feed the dynamic model with starting points coming from (static) pacing equilibria, we verify that the resulting (dynamic) pacing multipliers do not fluctuate much from the starting points.

\section{Equilibrium Analysis}\label{sec:eqmanalysis}

In this section we study the equilibria resulting from the pacing game. We first prove that all instances admit equilibria and later we study properties of these equilibria such as multiplicity and efficiency.

\subsection{Equilibrium Existence}

To characterize a pacing equilibrium, as introduced in Definition~\ref{def:pacing}, we require not
only a profile of strategies (where the $\alpha_i$ would correspond to
strategies) but also one of allocations. Even ignoring that we need allocations,
there are discontinuities involved that might be suspected to get in the way of
equilibrium existence: upon exceeding another bid there is a jump in one's
utility, and again for exceeding one's budget. On top of that, in the definition
of pacing equilibrium, we require buyers to break certain indifferences towards
higher bids: a buyer~$i$ who at $\alpha_i=1$ does not spend the budget is not allowed to use
a lower value of $\alpha_i$ in the definition. 
Despite these difficulties, we can show that a pacing equilibrium always exists via a smoothing argument.
This smoothing argument relies on a smoothed pacing game, where every good is split among all bids that lie within an additive band around the winning bid, with the split applied proportionally depending on where each bid falls in the band. This is reminiscent of how one might handle tie-breaking in practice (either in the indivisible case, or in the divisible case by dividing every divisible good up into many separate units): in the allocation rule, every bid has a small amount of noise added to it, before computing the allocation. (\citet{borgs2007dynamics} studied such a scheme, and show convergence results in the case of first price auctions with budgets.)

\begin{theorem}\label{th:existence}
Any pacing game admits a pacing equilibrium.
\end{theorem}

To provide this result we rely on a smoothed version of the pacing game, which
takes care of all the discontinuity issues. In the smoothed version, the
allocation varies continuously and is determined as a function of the $\alpha_i$
only, the penalty for exceeding one's budget varies continuously, and strict
incentive is given to bid higher. We show we can apply a pure Nash equilibrium
existence result to such games. We then show that if we take a sequence of such
games that converges to a (non-smoothed) pacing game, then this sequence of pure
Nash equilibria converges to a pacing equilibrium.

\begin{definition}
  For $\epsilon>0$ and $H>0$, an {\em $(\epsilon,H)$-smoothed pacing game} is
  a game where the set of pure strategies for each buyer $i$ is the set of pacing
  multipliers $\alpha_i\in [0,1]$. For a fixed choice of pacing multipliers,
  the original pacing auction market is modified as follows in order to compute
  allocations and payments:
\begin{itemize}
\item {\bf Reserve bid:} there is an artificial bid of $2\epsilon$ on all goods (treated as one of the buyers in the below).
\item {\bf Allocation and pricing rule:} For every good $j$, consider the highest bid $b_j^*=\max_i \alpha_i v_{ij}$.
Let $S_j = \{i: \alpha_i v_{ij} \geq b_j^* - \epsilon\}$ be the set of buyers close to the maximum bid for $j$.
Then $i \in S_j$ wins the following fraction of good $j$: $x_{ij} = \frac{\alpha_i v_{ij} - (b_j^* - \epsilon)}{\sum_{i' \in S_j} [\alpha_i v_{ij} - (b_j^* - \epsilon)]}$, and pays $s_{ij}=x_{ij}p_{j}$
for this, where $p_{j}$ is the highest bid on $j$ among buyers other than $i$, minus $\epsilon$ (which is necessarily at most $b_j^*-\epsilon$).
 For the other buyers, $x_{ij}=s_{ij}=0$.
\item %needed to give a bias towards higher alpha values
{\bf Additional artificial spend (to encourage higher bids from those who have not spent their budgets):}
Each buyer will additionally receive a quantity $\alpha_i$ of an artificial good (with unlimited supply) worth $2\epsilon$ per unit to her, and pay $\alpha_i \epsilon$ for this.
This results in a profit of $\alpha_i \epsilon$ if the budget is not exceeded by this payment.
\item {\bf Utility:} The utility of buyer $i$
is $(B_i - \alpha_i \epsilon - \sum_j s_{ij}) + 2\alpha_i \epsilon + \sum_j x_{ij}v_{ij}$
if she does not exceed the budget $B_i$, or
$H(B_i-\alpha_i \epsilon - \sum_j s_{ij}) + 2\alpha_i \epsilon + \sum_j x_{ij}v_{ij}$
if she exceeds it.
\end{itemize}
\label{def:smoothed}
\end{definition}

The smoothing of allocations and payments allows us to apply existence theorems
about pure-strategy Nash equilibria.

\begin{theorem}\label{th:smooth_exists}
  Consider a smoothed pacing game in which a strategy for buyer $i$ consists of choosing
  $\alpha_i \in [0,1]$. Also, let $M$ be any upper bound on
  the sum of a buyer's valuations in the game, including those for the artificial good.
   For $H > M/\epsilon$, the game admits a pure-strategy Nash equilibrium.
\end{theorem}

\begin{proof}{Proof.}
  We will apply a theorem
  by~\citet{debreu1952social}, \citet{glicksberg1952further}, and \citet{fan1952fixed}
  (see also \citealt[p.~20]{ozdaglar-GTcourse5})
that guarantees existence of a pure-strategy Nash equilibrium under the
following conditions (which we immediately show apply to our game):
\begin{itemize}
\item {\bf Compact and convex strategy space.} This holds because $\alpha_i \in
  [0,1]$.
\item {\bf Continuity of utility in all strategies.} This holds for the
  following reasons: $x_{ij}$ and $s_{ij}$ are continuous in all the
  $\alpha_{i'}$ (in particular, note that buyers $i$ who are just barely in
  $S_j$ with $\alpha_i v_{ij}=b_j^*-\epsilon$ receive $x_{ij}=0$). And utility
  is continuous in these quantities (in particular, note that the expressions
  for buyers who exceed and do not exceed the budget coincide at $2\alpha_i
  \epsilon + \sum_j x_{ij}v_{ij}$ when the budget is spent exactly).
\item {\bf Quasiconcavity of utility in the buyer's own strategy.} This means
  we must show that $u_i(\alpha_i,\alpha_{-i})$ is quasiconcave in $\alpha_i$.
  This is the case if there exists a number $t$ such that for $\alpha_i < t$,
  $u_i$ is nondecreasing in $\alpha_i$, and for $\alpha_i > t$, $u_i$ is
  nonincreasing in $\alpha_i$. Buyer $i$'s total spend $\alpha_i \epsilon +
  \sum_j s_{ij}$ is increasing and continuous in $\alpha_i$. Holding
  $\alpha_{-i}$ fixed, let $t$ be the value of $\alpha_i$ such that $\alpha_i
  \epsilon + \sum_j s_{ij} = B_i$ (if no such value exists we may set $t=1$).
  Then, for $\alpha_i < t$, $u_i$ is increasing in $\alpha_i$, because
  increasing $\alpha_i$ results in winning more goods (including more of the
  artificial good) at prices below $i$'s valuation ($\alpha_i$ does not affect
  $p_{j}$, and if $i$ is winning part of $j$ then $v_{ij} \geq \alpha v_{ij}
  \geq b_j^*-\epsilon \geq p_{j}$). For $\alpha_i > t$, $i$'s total spend
  (including on the artificial good) is increasing in $\alpha_i$, and any
  additional spend will exceed $i$'s budget, decreasing the utility term
  $H(B-\alpha_i \epsilon -\sum_j s_{ij})$ at rate $H$. Because each good
  (including the artificial good) costs at least $2\epsilon - \epsilon
  =\epsilon$, the value gained from goods bought increases at a rate of at most
  $M / \epsilon$, which by assumption is smaller. Hence, utility is decreasing
  in $\alpha_i$ when $\alpha_i>t$.
\end{itemize}
\hfill\Halmos
\end{proof}

With this result we are ready to prove Theorem~\ref{th:existence}.
Using the existence of pure-strategy Nash equilibria in smoothed pacing games,
we can show that a limit point of decreasingly smoothed games constitutes a
pacing equilibrium in the original pacing game.

\begin{proof}{Proof.}
For a given pacing game, consider a sequence of smoothed versions of it, defined by $(\epsilon^l,H^l)$, satisfying $H^l >  M/\epsilon^l$, $\lim_{l \rightarrow \infty} \epsilon^l=0$, and $\lim_{l \rightarrow \infty} H^l=\infty$.  Consider an associated sequence of equilibria of these games (guaranteed to exist by Theorem~\ref{th:smooth_exists}) defined by $\{\alpha_i^l, x_{ij}^l, p_{j}^l, s_{ij}^l\}$.  This sequence must have a subsequence with a limit point $\{\alpha_i^*, x_{ij}^*, p_{j}^*, s_{ij}^*\}$ by virtue of the fact
that these numbers lie in a compact space (the values provide an upper bound on the payments); replace the sequence by this subsequence.  We will show that this limit point is an equilibrium of the original pacing game, via the following claims.
\begin{itemize}

\item {\bf The allocation is feasible.} Since for each $l$ and $j$, $\sum_i x_{ij}^l \leq 1$, we must have $\sum_i x_{ij}^* \leq 1$.
Moreover, suppose that there exists $i$ with $v_{ij}>0$. Because $B_i > 0$, there is some positive value of $\alpha_i$ that guarantees $i$ stays below budget; hence $i$ will bid at least $\alpha_i v_{ij}$ for every $l$.  Thus, for sufficiently large $l$, $\epsilon^l$ will be sufficiently small that the reserve buyer wins none of $j$, and $\sum_{i'} x_{i'j}^l = 1$.  Hence $\sum_{i'} x_{i'j}^* = 1$ in this case.
Finally, if $x_{ij}^* > 0$, this implies that there exists $L$ such that for $l>L$, $\alpha_i^l v_{ij} \geq \max_{i'} \alpha_{i'}^l v_{i'j} - \epsilon^l$. Since $\lim_{l \rightarrow \infty} \epsilon^l=0$ this implies $\alpha_i^* v_{ij} \geq \max_{i'} \alpha_{i'}^* v_{i'j}$, so $i$ in fact is at least tied for the highest bid on $j$.

\item {\bf The payments are right.} $p_{j}^* = \lim_{l \rightarrow \infty} p_{j}^l$.  The latter is the highest other bid minus $\epsilon^l$.  The highest other bid converges to the highest other bid at the limit point (note the reserve bid goes to $0$), and $\epsilon^l$ goes to $0$.
Moreover, $s_{ij}^* = \lim_{l \rightarrow \infty} x_{ij}^l p_{j}^l = x_{ij}^* p_{j}^*$.

\item {\bf No buyer exceeds her budget.} We must show that for each buyer $i$,
  $\sum_j s_{ij}^* \leq B_i$.  Suppose not.  Then, there exists $\delta > 0$ such
  that for any $L$, we can find $l > L$ with $\sum_j s_{ij}^l \geq B_i +
  \delta$.  But if we let $L$ be such that for $l > L$, we have $H^l > M /
  \delta$, then the buyer's utility for the equilibrium of the resulting game $l$ is at most $M - \delta H^l < M-M = 0$.  (Spending on the artificial good only makes things worse.)  But the buyer can guarantee herself utility $0$ by setting $\alpha_i^l=0$, contradicting the fact that we have an equilibrium.  Hence no buyer exceeds her budget.

\item {\bf A buyer with $\alpha_i^*<1$ spends her entire budget.}  Suppose not, i.e., there is such a buyer with $\sum_j s_{ij}^* < B_i$.  Then we can find $L$ such that for $l > L$, both $\alpha_i^l \epsilon^l + \sum_j s_{ij}^l < B_i$ (because $\epsilon^l$ goes to $0$) and $\alpha_i^l < 1$.
But as we pointed out earlier, for such a buyer utility is strictly increasing in $\alpha_i^l$ (the strictness is due to the artificial good).  Thus this buyer is not best-responding, contradicting the fact that we have an equilibrium.  Hence a buyer
with $\alpha_i^*<1$ spends her entire budget.
\end{itemize}
\hfill\Halmos
\end{proof}

\subsection{Sensitivity and Multiplicity of Equilibria}\label{sec:multiple}

Knowing that at least one pacing equilibrium exists, we ask the following questions. First, can pacing equilibria be very sensitive to input parameters? Second, can a pacing game admit multiple pacing equilibria, and if so, can they differ significantly from each other?  We provide affirmative answers in each case.
For this, we need to quantify how different one equilibrium is from another. 
Fixing a feasible solution to a pacing game, we rely on the following three {\em objective functions} that capture instance-wide measures of interest.

\begin{definition}$~$\\
  $\bullet$ \textbf{Revenue} is the total spending in the game $(\sum_{ij} s_{ij})$.\\
  $\bullet$ \textbf{Social welfare} is the sum of winning valuations $(\sum_{ij} x_{ij} v_{ij})$.\\
  $\bullet$ \textbf{Paced welfare} is the sum of paced winning valuations $(\sum_{ij} x_{ij} \alpha_i v_{ij})$.
\end{definition}

Revenue and social welfare are natural objectives; we now justify why we consider paced welfare.  
If buyers' budgets are small, then their valuations are relevant only insofar as they indicate the {\em relative} values of the goods.  But they no longer make sense as an absolute dollar figure: if one were to double all the valuations, without touching the budget, nothing would change in the auctions.  The next observation makes this precise.

\begin{observation}
Given a pacing equilibrium where $\alpha_i < 1$ for some $i$, if we modify all of $i$'s valuations to $v'_{ij}=\beta_i v_{ij}$ where $\beta_i \geq \alpha_i$, then we can retain the original pacing equilibrium by setting $\alpha'_i = \alpha_i / \beta_i$. We call this an {\em irrelevant shift in valuations}.
\end{observation}

This leads us to a definition and a corresponding result.

\begin{definition}
A welfare measure is {\em robust to irrelevant shifts in valuations} if it produces the same value after an irrelevant shift in valuations.
A welfare measure {\em coincides with social welfare when budgets are large} if, whenever $\alpha_i=1$ for all buyers $i$, it evaluates to $\sum_{ij} x_{ij}v_{ij}$.
\end{definition}

\begin{proposition} \label{prop:characterization}
Paced welfare is the unique welfare measure that coincides with social welfare when budgets are large and is robust to irrelevant shifts in valuations.
\end{proposition}
\begin{proof}{Proof.}
It is straightforward to check that paced welfare satisfies the conditions.  To show that it does so uniquely, consider any welfare measure satisfying the two conditions and any feasible solution of a pacing game.  We prove that the welfare measure must coincide with paced welfare, by induction on the number of buyers $i$ with $\alpha_i<1$.  If there are $0$ such buyers, then this follows from the fact that the measure coincides with social welfare in this case.  Suppose we have shown it to be true with $k$ such buyers; we will show it with $k+1$.  Choose an arbitrary buyer $i$ with $\alpha_i<1$.  Modify the buyer's valuations to $v'_{ij}=\alpha_i v_{ij}$, and let $\alpha'_i = \alpha_i / \alpha_i = 1$.  This is an irrelevant shift in valuations, so the modification affects neither paced welfare nor the welfare measure under consideration.  But by the induction assumption, the two must coincide after the shift.  So they must have coincided before the shift as well.
\hfill\Halmos 
\end{proof}

Equipped with these objective functions, we look at concrete examples that show that equilibria are sensitive to budgets.
In particular, Examples~\ref{ex:paced welfare budget sensitive} and~\ref{ex:revenue budget sensitive} below show that small budget perturbations can cause large swings in paced welfare and revenue. Note that these examples admit a single equilibrium.
\begin{example} Large changes in objective function from small changes in budgets:
\begin{subexample} (Large paced welfare loss from small changes in budgets)
  \label{ex:paced welfare budget sensitive}
  Buyer $1$ has valuation $v_{11}=100$ and budget $B_1=1.01$. Buyer $2$ has valuation $v_{21}=1$ and budget $B_2=\infty$. Then we have a pacing equilibrium with $\alpha_1=\alpha_2=1$ where $1$ wins all of good $1$ for a \underline{\em paced welfare of $100$}. Moreover this is the unique pacing equilibrium because neither buyer can spend her whole budget.
  Now, reduce $B_1$ to $0.99$. We must still have $\alpha_2=1$. Hence, we must have $\alpha_1 \leq 0.01$, because otherwise $1$ will exceed her budget on good $1$. As a result, both buyers have a paced valuation of less than $1$, and thus, no matter the allocation, \underline{\em paced welfare is at most $1$}.
\end{subexample}
\begin{subexample} (Large revenue loss from small changes in budgets)
  \label{ex:revenue budget sensitive}
  Buyer $1$ has valuations $v_{11}=100$ and $v_{12}=100$, and budget $B_1=1.01$.  Buyer $2$ has valuations $v_{21}=1$ and $v_{22}=101$, and budget $B_2=\infty$.  Then we have a pacing equilibrium with $\alpha_1=\alpha_2=1$ where $1$ wins all of good $1$ at price $1$ and $2$ wins all of good $2$ at price $100$, for a \underline{\em total revenue of $101$}.  Moreover this is the unique pacing equilibrium: buyer $2$ cannot possibly spend his whole budget and hence must have $\alpha_2=1$, and given this, buyer $1$ cannot win any of good $2$ and will spend less than her whole budget on good $1$, so that $\alpha_1=1$ as well.
  Now, reduce $B_1$ to $0.99$.  We still must have $\alpha_2=1$.  Hence, we must have $\alpha_1 \leq 0.01$, because otherwise buyer $1$ will have to win all o good 1 and exceed her budget.  As a result, the second price on each good is less than one, and thus revenue from each good is at most $1$, for a \underline{\em total revenue of at most $2$}.
\end{subexample}
\end{example}

A second particularity of pacing equilibrium that we want to highlight is the possibility of \emph{multiplicity}, i.e., the existence of multiple equilibria, possibly substantially different in character.
Examples~\ref{ex:revenue_equilibria}, \ref{ex:welfare_equilibria_unpaced}, and \ref{ex:welfare_equilibria} given next, one after the other, show instances with multiple pacing equilibria. The examples focus on large difference in revenue, welfare, and paced welfare, respectively.

\begin{example} Two equilibria with large objective function difference:
\begin{subexample} (Two equilibria with large revenue difference) \label{ex:revenue_equilibria}
Let $v_{11}=v_{22}=100$, $v_{12}=v_{21}=1$, $v_{13}=v_{23}=99$,  and
$v_{14}=v_{34}=100$.  Let all other valuations be $0$. Moreover, let
buyers $1$ and $2$ have budget $1$ each, and let buyer $3$ have
budget $100$. One pacing equilibrium is
$\alpha_1=1$, $\alpha_2=0.01$, $\alpha_3=1$, where buyer $1$ wins good $1$
for $0.01$ and good $3$ for $0.99$, buyer $2$ wins good $2$ for $1$,
  and buyer $3$ wins good $4$ for $100$, resulting in a \underline{\em total revenue
  of $102$}. Another pacing equilibrium is
$\alpha_1=0.01$, $\alpha_2=1$, $\alpha_3=1$, where buyer $1$ wins good $1$
for $1$, buyer $2$ wins good $2$ for $0.01$ and good $3$ for $0.99$,
  and buyer $3$ wins good $4$ for $1$, resulting in a \underline{\em total revenue of
  $3$}.
\end{subexample}
\begin{subexample} (Two equilibria with large welfare difference) \label{ex:welfare_equilibria_unpaced}
Let $v_{11}=100$, $v_{22}=200$, $v_{12}=2$, $v_{21}=1$,
$v_{13}=v_{23}=99$, $v_{14}=0.01$, $v_{24}=1$, and $v_{34}=10000$.
Let all other valuations be $0$. Moreover, let $B_1=1$,  $B_2=2$, and
$B_3=0.01$.
One pacing equilibrium is $\alpha_1=1$, $\alpha_2=0.01$, $\alpha_3=1$,
where buyer $1$ wins good $1$ for $0.01$ and good $3$ for $0.99$,
buyer $2$ wins good $2$ for $2$, and buyer $3$ wins good $4$ for
  $0.01$, resulting in a \underline{\em total social welfare of $10399$}. Another
pacing equilibrium is $\alpha_1=0.01$, $\alpha_2=1$,
$\alpha_3=0.0001$, where buyer $1$ wins good $1$ for $1$; buyer $2$
wins good $2$ for $0.02$, good $3$ for $0.99$, and a fraction $0.99$
of good $4$ at $0.99$; and buyer $4$ wins a fraction $0.01$ of good
  $4$ at $0.01$.  This results in a \underline{\em total social welfare of $499.99$}.
\end{subexample}

\begin{subexample} (Two equilibria with large paced welfare difference)\label{ex:welfare_equilibria}
Let $v_{11}=v_{22}=100$, $v_{12}=v_{21}=1$, $v_{13}=v_{23}=99$,
$v_{14}=10000$, and $v_{24}=0$. Moreover, let buyers $1$ and $2$ have
budget $1$ each. One pacing equilibrium is $\alpha_1=1$, $\alpha_2=0.01$,
where buyer $1$ wins good $1$ for $0.01$, good $3$ for $0.99$, and
good $4$ for $0$, and buyer $2$ wins good $2$ for $1$, resulting in a
  \underline{\em total paced welfare of $100+99+10000+1=10200$}. Another pacing
equilibrium is $\alpha_1=0.01,\alpha_2=1$, where buyer $1$ wins good
$1$ for $1$ and good $4$ for $0$, and buyer $2$ wins good $2$ for
  $0.01$ and good $3$ for $0.99$, resulting in a \underline{\em total paced welfare of
  $1+100+100+99=300$}.
\end{subexample}
\end{example}

The last examples would seem to suggest that in practice it may be worthwhile to consider equilibrium selection procedures. 
We highlight that while the multiple equilibria in the examples have very different objective values, multiplicity might not happen often in practice, and the gaps may not be that large when it does.
That conclusion is driven from a computational study, presented in Section~\ref{sec:experiments}, which investigates how often multiplicity happens and how large the gaps are.

In work related to ours, \citet{balseiro2015repeated}, \citet{balseiro2021budget}, and \citet{balseiro2017dynamic} study stochastic versions of a pacing game, where the buyers draw their valuations independently, typically from a well-behaved CDF that is absolutely continuous and potentially bounded. The budget constraint is then typically only required to hold in expectation (\citet{balseiro2017dynamic} is an exception since they consider an adaptive dynamic setting).
The models put forward by those papers do not admit multiple equilibria, which motivates the question of why our model does. One potential conjecture would be that it is the non-smoothness of our discrete setting that enables the multiplicity.
We provide evidence against this conjecture by providing examples where valuations for a good are sampled, and the budget constraint is required to hold in expectation. First, we observe that a simple extension of one of our instances with multiplicity allows us to show that multiple equilibria exist even in the case of valuations drawn from an absolutely continuous CDF with bounded density (a smooth setting with correlated valuations). This also shows that even smooth instances with a continuum of goods exhibit multiplicity. Second, we show an example where two buyers with unequal budgets sample valuations uniformly and independently drawn from a simple discrete distribution (an iid valuation setting with a discrete distribution). Both examples are explained in Appendix~\ref{sec:stoch-valuations} in the e-companion.

Our examples show that neither absolute continuity of the CDF (or equivalently a continuum of goods), nor having iid valuations, is enough to guarantee uniqueness of pacing equilibria by itself. Based on our work, it seems plausible that if we assume both absolutely continuous CDFs with bounded density \emph{and} independence of valuations, then there is potentially a single unique pacing equilibrium. These findings are supported by those of \citet{balseiro2021budget}. They show that there is a unique equilibrium in several 2-buyer settings with independent exponential, uniform, Rayleigh, and Weibull distributions (their results do also use a non-zero reserve price and are thus slightly incomparable). They also show multiplicity in a 2-buyer example where valuations are independent from piecewise-linear CDFs (again these results utilize a reserve price as well). These results were added to the 2018 working version of their paper, which cites an earlier version of the present paper as showing that multiplicative pacing can lead to multiple and unstable equilibria.

Finally, let us go back to the interpretation of pacing equilibria as games between proxy bidders, assuming that the platform has full information because advertisers truthfully report their values and budgets. This information structure assumes that an advertiser may not be better off misreporting to strategically get to a better pacing equilibrium when the proxy bidders play the game. 
The following example shows that it is possible that an advertiser achieves a large gain in utility through a small change in reported values. 
In Section~\ref{sec:experiments}, we will empirically test whether practically advertisers have an incentive to misreport, 
and show that when there is competition our games do not create large incentive issues in practice.

\begin{example}[Large utility gain when slightly misreporting values]
\label{ex:misreporting}
Buyer $1$ has valuations $v_{11}=100$ and $v_{12}=100$, and budget $B_1=0.99$.  Buyer $2$ has valuations $v_{21}=0.98$ and $v_{22}=101$, and budget $B_2=\infty$.  Then we have a pacing equilibrium with $\alpha_1=\alpha_2=1$, where buyer $1$ wins all of good $1$ at price $0.98$ and Buyer $2$ wins all of good $2$ at price $100$. \underline{\em Buyer $2$'s utility for this outcome is $101-100=1$}. Moreover this is the unique pacing equilibrium: buyer $2$ cannot possibly spend his whole budget and hence must have $\alpha_2=1$, and given this, buyer $1$ cannot win any of good $2$ and will spend less than her whole budget on good $1$, so that $\alpha_1=1$ as well.
Now, increase the \emph{reported} $v_{21}$ to $1$.  We still must have $\alpha_2=1$, since buyer $2$ has infinite budget.  Hence, we must have $\alpha_1 \leq 0.01$, because otherwise buyer $1$ will exceed her budget on good $1$.  As a result, buyer $2$ wins all of good $2$, receiving value $101$ at a price no larger than $1$; buyer $2$ also wins some of good $1$, receiving nonnegative value and cost at most $1$. \underline{\em Buyer $2$'s utility for this outcome is at least $99$}.
\end{example}

\section{Relationship to Competitive Equilibrium}
\label{sec:competitive_equilibrium}

We now show that pacing equilibria are a refinement of competitive (Walrasian)
equilibria, a widely studied concept for understanding markets. These results
are in contrast to those for stochastic settings in
\citet{balseiro2021budget} and \citet{balseiro2017dynamic}, which do not have a
such a relationship to competitive equilibria. We define a competitive
equilibrium with budgets as follows.

\begin{definition}
  A {\em competitive equilibrium with budgets} consists of a price $p_j$ on every good $j$, and an
  allocation of goods to buyers such that every buyer buys a bundle that maximizes her utility, subject to her budget constraint.
  (A buyer is allowed to acquire goods partially.)  That is, buyer $i$'s bundle, consisting of fractions $\{x_{ij}\}$ that she obtains of each good $j$,
   must be in $\arg \max_{\{0 \leq x_{ij}\leq 1: \sum_j x_{ij} p_{j} \leq B_i\}} \{\sum_j x_{ij}(v_{ij}-p_{j}) \}$. Additionally, every good with a positive price must be fully allocated.
\end{definition}

Competitive equilibria have several attractive properties. For example, in a competitive equilibrium, buyers have \emph{no envy}, meaning that they prefer their own bundle to that of any other buyer (in a budget-adjusted sense when budgets are not equal). 

We can characterize the optimal actions of buyers as selecting goods in decreasing order of bang-per-buck. This will be helpful in the derivations below.
\begin{proposition}
  A bundle maximizes a buyer's utility under her budget constraint if and only if 
  she buys (parts of) goods in decreasing order of bang-per-buck ($v_{ij} / p_{j}$), starting with the highest, until she either runs out of budget or reaches goods such that $v_{ij}<p_{j}$.
\end{proposition}

\begin{proposition}\label{prop:pacing_to_ce}
  For every pacing equilibrium, there is an equivalent competitive equilibrium.
\end{proposition}

\begin{proof}{Proof.}
Given the pacing equilibrium, set the price of each good equal to the second-highest paced bid on it (possibly equal to the highest bid), and use the same allocation as in the pacing equilibrium.
Note this means buyers also pay the same as in the pacing equilibrium.
Every buyer $i$ that does not run out of budget (and therefore has multiplier $\alpha_i=1$) buys every good $j$ with $v_{ij}>p_{j}$, because the valuation being above the price means that buyer was uniquely the highest bidder on it in the pacing equilibrium; and buys no good $j$ with $v_{ij}<p_{j}$, because the valuation being below the price means that the buyer was not a highest bidder for it.
A buyer that runs out of budget is spending her money on maximum bang-per-buck goods, because $x_{ij}=1$ for every good for which $\alpha_iv_{ij} > p_j \iff \frac{v_{ij}}{p_j} > \frac{1}{\alpha}_i \geq 1$.
For a good $j$ such that $\alpha_iv_{ij} = p_j \iff \frac{v_{ij}}{p_j} = \frac{1}{\alpha}_i \geq 1$ she may buy a fractional amount: such goods provide worse bang-per-buck than any other good that $i$ buys, and thus since they are spending their whole budget they do not wish to buy more of $j$.
No buyer buys anything that is priced above her valuation, because the price being above her valuation means that she did not have the highest (paced) bid on that good in the pacing equilibrium.
\hfill\Halmos
\end{proof}

Proposition~\ref{prop:pacing_to_ce} shows that SPPE inherits all the desirable properties of competitive equilibrium, which are quite desirable to buyers, and thus also to the platform. For example, it follows that SPPE is guaranteed to have no envy among buyers (in a budget-adjusted sense), while being Pareto optimal. Another interesting consequence is that SPPE inherits \emph{strategyproof in the large} properties from competitive equilibrium~\citep{azevedo2019strategy,kroer2019scalable}.

The converse is not true: pacing equilibria strictly refine competitive equilibria.
For example, consider a setting with a single buyer and good, with
value $v_{11}=1$. All pacing equilibria have zero revenue, but a competitive
equilibrium can have $p_{1}=\frac{1}{2}$.
Hence, a competitive equilibrium can result in higher revenue than any pacing equilibrium.
The opposite direction is more interesting: a competitive equilibrium can yield a {\em lower} revenue than any pacing
equilibrium. The intuition is that setting a high price on one good
can drain some buyer's budget, thereby making that buyer effectively ``paced,''
as shown below.
\begin{example}[Competitive equilibrium with more revenue than pacing equilibrium]
  Suppose we have 3~buyers and 3~goods. 
  Buyer~1 values good~1 at $101$, buyer~2 values goods 1, 2 and 3 at $100$, $200$, and $10$, respectively, and buyer~3 values good~3
  at $1$. All other valuations are $0$. 
  Buyer~2 has budget $10.1$, the other two have budget $\infty$. 
  Since buyer~2 faces no competition for good~2, in a pacing equilibrium, buyer~2 gets it for free and will pay at most $1$.
  Hence, no buyer will be paced, resulting in independent second price
  auctions. \underline{\em The revenue for good~1 is $100$}. However, in a
  competitive equilibrium, we can arbitrarily set a price of $10$ for good~2.
  We then price good~3 at $1$ and let buyer~2 buy one tenth of it,
  thereby spending his budget. Finally, we price good~1 at $101$ so buyer~2
  will no longer want to buy it. (For buyer~2, the goods ordered by bang per
  buck are 2, 3 and 1, which satisfies the competitive equilibrium conditions.) 
  \underline{\em Revenue has plummeted to $11+10+1=22$}.
  \label{ex:ce_lower_rev}
\end{example}
Nonetheless, every competitive equilibrium can be reinterpreted as a pacing equilibrium as well.
\begin{proposition}
  For every competitive equilibrium, one can construct an equivalent pacing equilibrium
  after possibly adding a single buyer who acts as a price setter but who does not win anything.
  \label{prop:ce_to_pacing}
\end{proposition}

\begin{proof}{Proof.}
  Given the competitive equilibrium, add a buyer with infinite budget who bids exactly $p_j$ (as in the
  competitive equilibrium) on every good. Use the same allocation as in the
  competitive equilibrium (so the new buyer wins nothing). Buyers who bought every good for which their
  valuations exceeded the price are not paced. Buyers who ran out of budget are
  paced as follows. Since they bought goods in order of maximum bang-per-buck,
  for each such buyer $i$, consider the good $j$ with minimum $\frac{v_{ij}}{p_j}$
  of which she still bought some. Define $\alpha_i = p_j/v_{ij}$ for that good.

  We must show that every good is in fact won by the highest paced buyer for it
  and that the added buyer is always the second highest (allowing for ties).
  First, we show that the added buyer is never the uniquely highest bid, because
  its bids are always (weakly) exceeded by any buyer who wins (some of) the good
  in the competitive equilibrium. If that buyer is an unpaced buyer, we must
  have $v_{ij} \geq p_j$, because otherwise she would not have bought the good in
  the competitive equilibrium. If it is a paced buyer, because she buys some of
  $j$ in the competitive equilibrium, it follows that $\frac{v_{ij}}{p_j} \ge
  \frac{1}{\alpha_i}$ by the definition of $\alpha_i$. Then,
  $\alpha_iv_{ij} \geq p_j$.

  Next, we show that there cannot be two or more buyers with paced bids strictly
  higher than that of the added buyer. For suppose there are; there is at least one
  that will not win the entire good. If that buyer is not paced, then we have
  $v_{ij}>p_j$, but this leads to a contradiction because unpaced buyers must
  have won all such goods completely in the competitive equilibrium. If the buyer is
  paced, we have $\alpha_iv_{ij} > p_j \iff \frac{v_{ij}}{p_j} >
  \frac{1}{\alpha_i}$. By the definition of $\alpha_i$ that means there is some other
  good $j'$ with $\frac{v_{ij'}}{p_j'} = \frac{1}{\alpha_i} <
  \frac{v_{ij}}{p_j}$ of which $i$ bought some in the competitive equilibrium. But
  this leads to a contradiction, because if so, then $i$ should have bought all of
  $j$ in the competitive equilibrium before moving on to $j'$.

  It follows that every buyer winning part of a good has the highest paced bid
  on that good and the added buyer is always (possibly tied for) second. This
  means that the allocation and prices are consistent with the definition of
  pacing equilibrium.
\hfill\Halmos
\end{proof}

While Proposition~\ref{prop:ce_to_pacing} shows that every competitive equilibrium can be implemented by adding an additional price-setting buyer, this proposition also has another interesting interpretation: the price-setting buyer could also be implemented via reserve prices. Thus, any competitive equilibrium can be implemented by constructing an SPPE with reserve prices.

Combining Example~\ref{ex:ce_lower_rev} with
Proposition~\ref{prop:ce_to_pacing}, we obtain a revenue nonmonotonicity result.
\begin{corollary}
  Adding a buyer may decrease the revenue at a pacing equilibrium.
\end{corollary}

Finally, the {\em first fundamental theorem of welfare economics} states that
competitive equilibria are Pareto optimal. Although this is a
known result, we include a direct proof for our setting in Appendix~\ref{sec:missingproofs} in the e-companion.
This, together with Proposition~\ref{prop:pacing_to_ce}, implies that pacing equilibria are Pareto optimal as well.
\begin{proposition}\label{prop:ce_pareto}
Any pacing equilibrium is Pareto optimal (when considering the utilities of both the seller and buyers).
\label{prop:pacing_pareto}
\end{proposition}

\section{Computing Pacing Equilibria}
\label{sec:hardness}

In this section, we look at algorithmic issues regarding computation of equilibria. We start with some complexity and hardness results, continue exploring whether iterating best responses will be useful to find equilibria, and present a MIP formulation that can helps us find equilibria.

\subsection{Complexity Results}

Motivated by equilibrium existence, and having defined the relevant objectives, 
we investigate the complexity of computing an equilibrium that optimizes an objective.
(We leave open the complexity of identifying an arbitrary equilibrium.)
Using a pacing equilibrium gadget that captures binary variables, we
can reduce {\sc 3SAT} to our problem.
An instance of {\sc 3SAT} consists of a tuple $(V,C)$, where $V$ is a set of Boolean
variables, and $C$ is a set of clauses of the form $\left(l_1 \lor l_2 \lor
l_3\right)$ with $l_i$ representing literals.
We define the decision versions of our problems and show hardness results for them.

\begin{definition}
  We are given goods, buyers, buyers' valuations for goods, buyers' budgets,
  and a number $T$. 
  {\sc MAX-REVENUE-PACING} consists in deciding 
  whether there exists a pacing equilibrium that
  achieves revenue at least $T$.
  {\sc MAX-WELFARE-PACING} and {\sc MAX-PACED-WELFARE-PACING} are similar but for
  social welfare, and 
  paced social welfare, respectively.
\end{definition}

\begin{theorem}
\label{th:max_rev}
  {\sc MAX-REVENUE-PACING}, {\sc MAX-WELFARE-PACING} and {\sc MAX-PACED-WELFARE-PACING} are NP-complete.
\end{theorem}

While full proofs are deferred to Appendix~\ref{sec:missingproofs} in the e-companion because they are technical, we provide the intuition for the
proof here. To get the results, we rely on Example~\ref{ex:binary_choice},
given below. This is an auction-market instance that models binary decisions.
We use one instance for each
variable, with both buyers representing literals true and false. Given a 
{\sc 3SAT} instance, we construct an auction market in which
additional buyers and the objectives encode whether all clauses are satisfied.

\begin{examplee}[Gadget for binary decisions]
  Given $K_1>0$, $\alpha>0$, $\delta \geq 0$ (with $\alpha+\delta<1$), and small
  $\epsilon$, let $K_2=\frac{1-\alpha-\delta}{2\alpha}K_1$. Let
  $v_{11}=v_{12}=v_{21}=v_{22}=K_2$, $v_{23}=v_{14}=K_1$, and
  $v_{13}=v_{24}=K_1/\alpha + \epsilon$. Both buyers have budget $K_1$. One
  pacing equilibrium is $\alpha_1=1$, $\alpha_2=\alpha$. This results in buyer $1$
  winning goods $1$, $2$, and $3$, for a total price of $2\alpha K_2 + \alpha
  K_1 = (1-\alpha-\delta)K_1+\alpha K_1 = (1-\delta) K_1$, and buyer $2$
  winning good $4$ for a total price of $K_1$. By symmetry, there is another
  equilibrium with $\alpha_1=\alpha$, $\alpha_2=1$, in which buyer $2$ retains
  $\delta K_1$ of his budget.
\label{ex:binary_choice}
\end{examplee}

For small $\alpha$ and $\delta$, this instance does not admit a 
pacing equilibrium where both buyers have even a moderately high
multiplier. Hence, if we were interested in pacing equilibria
with high multipliers, we can choose to make either $\alpha_1$ or $\alpha_2$ as
high as possible, but we cannot attempt to make both of them somewhat high at the
same time.

These hardness results limit the performance that we may expect from simple dynamics.
Hence, it may be worthwhile to attempt to intelligently guide the dynamics to improve the chances of ending up at a desirable equilibrium.

A natural question to ask is whether it is possible to approximately maximize any of these objectives that lead to NP-complete problems. This would be a very desirable result. However, an approximation algorithm would need to output a valid pacing equilibrium. Whether an arbitrary pacing equilibrium can be computed in polynomial time is itself an interesting open problem. We believe that computing a pacing equilibrium could be a PPAD-complete problem, which would preclude approximation algorithms.

\subsection{Iterated Best Responses}
A standard approach to find the equilibria of a game is to iterate best responses.  In many games, this is a reasonable learning procedure that converges to an equilibrium.
However, for pacing games this is not guaranteed to work if we start from an inadequate solution. 
Indeed, if we sequentially set each buyer's bids via best-responding multiplicatively paced bidding, we may end up with a cycling sequence of pacing vectors, as the following example demonstrates.

\begin{table}
  \caption{Valuations resulting in cycling best responses}
  \label{fig:cycling}
  \centering
  \small
  \begin{tabular}{rrrrrrr}
    $i$ & $v_{i,1}$ & $v_{i,2}$ & $v_{i,3}$ & $v_{i,4}$ & $v_{i,5}$ & $v_{i,6}$\\
    \hline
1& 100.0  & 1300.0 & 123.0  & 0.0    & 11.0   & 0.0    \\
2& 0.0    & 6503.0 & 300.6  & 501.0  & 0.0    & 25.0   \\
3& 50.0 & 0.0  & 0.0  & 500.0& 10.0 & 5.0  
  \end{tabular}
\end{table}

\begin{figure}
  \centering
	\includegraphics[width=0.8\columnwidth]{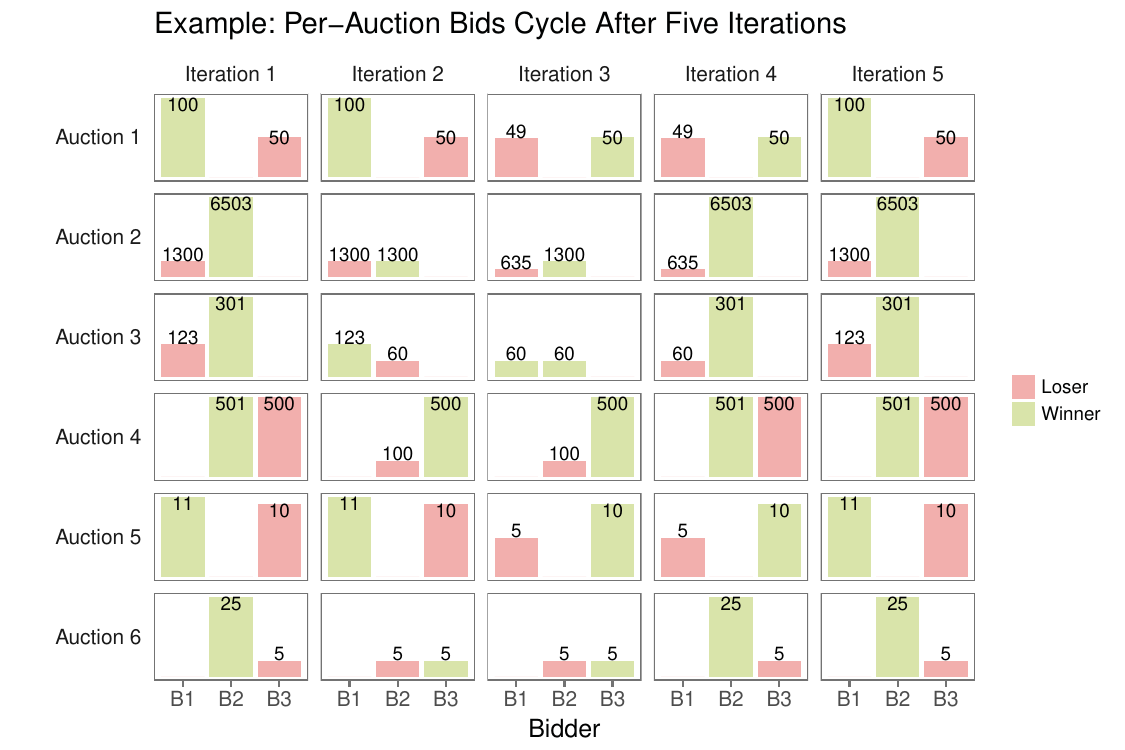}
	\caption{Best-response bids for the cycling example}
	\label{fig:examplebids}
\end{figure}

\begin{examplee}[Best responses may cycle] \label{ex:bids_appendix}
Consider an instance given by the set of valuations shown in Table~\ref{fig:cycling} and budgets $60$, $1300$ and $\infty$, for buyers $1$ to $3$, respectively. 
All buyers start with pacing multipliers equal to 1. 
Iterating best responses, all multipliers return to 1 after 5 iterations. 
See Figure~\ref{fig:examplebids} for an illustration and Appendix~\ref{sec:cycling} in the e-companion for a step-by-step explanation.
\end{examplee}

Although applying iterated best responses to pacing multipliers can cycle, the example above still admits multipliers that constitute an equilibrium with the corresponding fractional allocation. Since we know that an equilibrium exists and iterating best responses will not work as a general method to find one, we need another way to find an equilibrium. The next section develops such a method, while Appendix~\ref{sec:brdynamics} in the e-companion contains experiments showing the practical performance of best-response dynamics.

\subsection{MIP Formulation of Pacing Equilibria}
\label{sec:MIP}

Even though in earlier sections we showed that computing equilibria is hard in the worst case and that iterating best responses may cycle,
this does not mean that it is a hard problem in practice and for specific instances.
Being able to compute equilibria will allow us to study their
properties (e.g., find gaps among multiple equilibria, study incentive compatibility), and to use them as initial solutions when learning pacing multipliers in dynamic settings.
We provide a MIP formulation in which the constraints are equivalent to the equilibrium conditions. This guarantees that
a solution is feasible if and only if it satisfies the conditions given in Definition~\ref{def:pacing}. By optimizing with respect to various objectives,
we can refine the solution procedure and find different equilibria.
To define the problem,
it will be convenient to let $\bar{v}_{j} = \max_{i \in N}v_{i,j}$ be the
maximum value for good $j$ for any buyer.
We will need the following variables:
$~$\\[2mm]
\begin{minipage}{0.432\textwidth}
\begin{itemize}
	\item $\alpha_i \in [0, 1]$ : Buyer $i$'s pacing multiplier. 
	\item $s_{ij} \in \Rplus$ : Buyer $i$'s spend on good $j$. 
	\item $p_j \in \Rplus$ : Price of good $j$. 
	\item $h_{j} \in \Rplus$ : The highest bid for good $j$.
\end{itemize}
\end{minipage}
\hfill
\begin{minipage}{0.582\textwidth}
\begin{itemize}
	\item $d_{ij} \in \{0, 1\}$ : 1 if buyer $i$ may win any part of good $j$. 
	\item $y_{i} \in \{0, 1\}$ : 1 if buyer $i$ spends its full budget. 
	\item $w_{ij} \in \{0, 1\}$ : 1 if buyer $i$ is the winner of good $j$. 
	\item $r_{ij} \in \{0, 1\}$ : 1 if buyer $i$ is the second price for good $j$. 
\end{itemize}
\end{minipage}

\bigskip

Most variables are self-explanatory, as they denote the same as in the
pacing-game definition. Variables $w_{ij}$, and $r_{ij}$ represent a buyer that is considered {\em the
winner} and a buyer that is considered {\em the runner up} because the bid was a second price, respectively, for each good $j$.
The winner does not participate in lower-bounding the price (constraint~\eqref{eqn:price}), and the runner up upper bounds the price (constraint~\eqref{eqn:priceupper}).
In both cases, ties are broken arbitrarily but only one buyer can be chosen.
Although there could be multiple winners and runner-ups,
selecting exactly one of them is useful to encode the rules of a second price auction.

The equilibria of the pacing game are given exactly by the feasible solutions to the following MIP.
From a feasible solution, we get
pacing multipliers $\alpha_i$ for each buyer and spendings $s_{ij}$ for each buyer-good pair. The
fraction of good~$j$ allocated to buyer $i$ can then be computed as
$x_{ij}=s_{ij} / p_j$. (This last computation is not done inside the MIP because it would be nonlinear, but it is an easy computation to do once a solution to the MIP is obtained.)

\begin{centering}
\begin{minipage}{0.42\textwidth}
\begin{alignat}{2}
& \sum_{j \in M} s_{ij}  \leq  B_i &\ & (\forall i \!\in\! N) \label{eqn:budget}  \\
&\sum_{j \in M} s_{ij} \geq y_i B_i &\ & (\forall i \!\in\! N) \label{eqn:ydef} \\
& \alpha_i \geq 1 - y_i &\ & (\forall i \!\in\! N) \label{eqn:alpha} \\
& \sum_{i \in N} s_{ij} = p_j &\ & (\forall j \!\in\! M) \label{eqn:spend} \\
& s_{ij} \leq B_i d_{ij} &\ & (\forall i \!\in\! N, j \!\in\! M) \label{eqn:xdef} \\
& h_j \geq \alpha_i v_{ij} &\ & (\forall i \!\in\! N, j \!\in\! M) \label{eqn:highbid}
\end{alignat}
\end{minipage}
\hfill
\begin{minipage}{0.47\textwidth}
  \begin{alignat}{2}
& h_j \leq \alpha_i v_{ij} + (1 - d_{ij}) \bar{v}_j &\ & (\forall i \!\in\! N, j \!\in\! M) \label{eqn:highbidupper} \\
& w_{ij} \leq d_{ij} &\ & (\forall i \!\in\! N, j \!\in\! M) \label{eqn:winner} \\
& p_j \geq \alpha_i v_{ij} - w_{ij} v_{ij} &\ & (\forall i \!\in\! N, j \!\in\! M) \label{eqn:price} \\
& p_j \leq \alpha_i v_{ij} + (1 - r_{ij}) \bar{v}_j &\ & (\forall i \!\in\! N, j \!\in\! M) \label{eqn:priceupper} \\
& \sum_{i \in N} w_{ij} = 1 &\ & (\forall j \!\in\! M) \label{eqn:onewinner} \\
& \sum_{i \in N} r_{ij} = 1 &\ & (\forall j \!\in\! M) \label{eqn:onerunnerup} \\
& r_{ij} + w_{ij} \leq 1 &\ & (\forall i \!\in\! N, j \!\in\! M) \label{eqn:distinctwr}
\end{alignat}
\end{minipage}
\end{centering}
$~$\\[2mm]

We now describe the constraints.
Constraint~\eqref{eqn:budget} ensures that a buyer can spend no more than its
budget, while
\eqref{eqn:ydef} ensures that a buyer's total spend must be at least as
large as its budget if that buyer is spending its full budget (this enforces
the definition of $y_i$).
Constraint~\eqref{eqn:alpha} ensures that a buyer must have a pacing multiplier of
at least $1$ if it does not spend its full budget, 
\eqref{eqn:spend} ensures that the total spend of a good across buyers
must equal the price of that good, and
\eqref{eqn:xdef} ensures that a buyer's spend on a good is no greater than 0 if it
did not win part of that good.
Constraint~\eqref{eqn:highbid} ensures that the highest bid for a good must be at least as high as every paced bid for that good, and
\eqref{eqn:highbidupper} ensures that the highest bid for a good must be no greater than the paced bid of every buyer that wins part of that good.
Constraint~\eqref{eqn:winner} ensures that the designated winner for a good is designated as allowed to win a partial amount of that good, and
\eqref{eqn:price} ensures that the price for a good is at least as high as all paced bids besides the designated winner's paced bid.
Constraint~\eqref{eqn:priceupper} ensures that the price for a good is no greater than the runner-up's paced bid,
\eqref{eqn:onewinner} ensures that there is exactly one designated winner, 
\eqref{eqn:onerunnerup} ensures that there is exactly one designated runner-up, and
\eqref{eqn:distinctwr} ensures that a buyer cannot be both the designated winner and the designated runner-up of a given auction.

A revenue-maximizing pacing equilibrium can be computed by maximizing 
$\sum_{j \in M} p_j$ in the feasible region defined above,
whereas one can use $\max \sum_{j \in M} h_j$ to maximize the sum of the winning
paced bids.

We show in Appendix~\ref{sec:appendix_proofMIP} in the e-companion that our MIP correctly computes a pacing equilibrium.
\begin{proposition}\label{prop:correct_formulation}
  A solution to the MIP 
  \eqref{eqn:budget}-\eqref{eqn:distinctwr} is feasible if and only if it corresponds to the conditions of a pacing
  equilibrium.
\end{proposition}

If we are not concerned with a particular objective, but instead just want to compute
any one pacing equilibrium, we can use the following two approaches: The first
is to simply run the original MIP as a feasibility problem with no objective. 
The second is to relax the complementarity
condition \eqref{eqn:alpha}. We introduce a variable $z_i$ for
each buyer $i$ that represents whether that buyer satisfies
\eqref{eqn:alpha}. We replace \eqref{eqn:alpha} by
  $\alpha_i \geq 1 - y_i - z_i (\forall i \in N) $.
If $z_i=1$, then this constraint is no longer active
since $\alpha_i \geq 0 \geq -y_i$ is implied by the nonnegativity of $\alpha_i$ and
$y_i$. If $z_0=0$ then this constraint is our
standard complementarity condition on $\alpha_i$ and $y_i$. We can then solve
this relaxed MIP with the objective $\sum_{i\in N}z_i$. A solution where
the objective is zero corresponds to a feasible solution to the original MIP.

\newcommand{\algfont}{\sc}
\newcommand{\algname}[1]{{\algfont{#1}}\xspace}
\newcommand{\AdaptiveAlg}{\algname{AdaptivePacing}}
\newcommand{\ScaleUp}{\algname{ScaleInstance}}

\section{Computational Experiments}\label{sec:experiments}

In this section, we revisit our analytical results from an empirical point of view and put those results in perspective through a computational study. Rather than investigating worst-case instances as before, we consider various \emph{distributions} over pacing instances that attempt to capture real-life phenomena.
We investigate the following questions.
\begin{description}
\item[MIP Scalability.] 
  How large are the instances that the MIP can solve?
    We formulated a MIP to compute the value-maximizing pacing equilibria, but the problem of computing the value-maximizing pacing equilibrium is NP-complete.
\item[Equilibrium Analysis.] What are the empirical properties of pacing equilibria? 
  We showed that pacing equilibria are guaranteed to exist, but they are not necessarily unique, and there can be large gaps between the highest- and lowest-valued equilibria with respect to revenue, welfare, and paced welfare.
\item[Incentive Compatibility.] Does the system provide the right incentives?
    The pacing system takes advertisers' reported values as input. While an equilibrium is guaranteed to exist for those values, we showed that a buyer can sometimes increase their utility by misreporting their values. This calls for a deeper study of the welfare properties of pacing equilibria to understand when input values are truthful.
  \item[A Dynamic Setup.] How can our analytical results lead to improvements of practical pacing algorithms? We study a static game but our framework can inform how one paces budgets in a dynamic setting with noisy realizations of impressions.
\end{description}

The remainder of this section is organized as follows. In Section~\ref{sec:problem_instances}, we describe the different classes of problem instances we construct. In Section~\ref{sec:mip_scalability}, we describe how well the MIP scales on these different instances. Section~\ref{sec:empirical_eqa} describes equilibrium properties: empirical gaps between equilibria, and incentives for advertisers to misreport bids and budgets. Section~\ref{sec:improve_heuristics} explores using the MIP to seed a heuristic algorithm for a dynamic setup of the problem.

\subsection{Problem Instances} \label{sec:problem_instances}

We run experiments on two types of problem instances: \emph{stylized} instances, which were generated from a distribution over bipartite graphs; and \emph{realistic} instances, for which a bipartite graph was constructed from real-world auction markets. 
We describe how each type of instance is constructed below. 
Recall that a pacing instance consists of a tuple $(n, m, (v_{ij})_{i \in N, j \in M}, (B_i)_{i \in N})$, where $n$ is the number of buyers, $m$ is the number of goods (we use `auctions' interchangeably), $v_{ij}$ is buyer $i$'s value for winning good $j$, and $B_i$ is buyer $i$'s budget. 
We denote $N=\{1,\ldots,n\}$ and similarly with $M$ and $m$.

\paragraph{Stylized Instances} For stylized instances, we consider three distributions over bipartite graphs: \emph{complete}, \emph{sampled}, and \emph{correlated}. They refer to how the graph is connected and the correlation between edge weights. 

\begin{description}
\item[Complete.]
  In complete graph instances, every buyer is interested in every good. 
For each buyer~$i$ and good $j$, the valuation $v_{ij}$ is drawn uniformly
    iid from $[0,1]$. For each buyer~$i$, its budget $B_i$ is drawn uniformly from $[0, \sum_{j = 1}^m v_{ij}/n]$.

\item[Sampled.]
  Sampled graph instances are generated similarly to complete graphs, except that buyers are interested in a subset of goods. 
For each good, a subset of interested buyers is sampled uniformly at random from the power set of $\{1,\ldots,n\}$. 
If a buyer happens to not be interested in any goods, a single good of interest is uniformly sampled for that buyer. 
    Valuations $v_{ij}$ for the resulting edges and budgets $B_i$ are generated in the same manner as for complete graph instances.

\item[Correlated.]
  Correlated graphs are similar to sampled graphs, except that, for each good, the valuations are correlated across buyers through
the additional parameter~$\sigma$. For each good $j$, an expected valuation $\mu_j$ is drawn uniformly at random from $[0,1]$. For each buyer-good pair, valuation $v_{ij}$
    is then sampled from a Gaussian distribution truncated to $[0,1]$ with mean $\mu_j$ and standard deviation $\sigma$.
\end{description}

We generated 5 stylized instances for all combinations of instance type $\in \{$complete, sampled, correlated$\}$, numbers of buyers $n \in \{2, 4, 6, 8, 10\}$, number of goods $m \in \{4, 6, 8, 10, 11, 12, 14\}$, and in the case of correlated instances, standard deviations $\sigma \in \{0.01, \ldots, 0.09, 0.1, 0.2, 0.3\}$.
This resulted in a total of 175 complete instances, 175 sampled instances, and 2100 correlated instances, totalling 2450 stylized instances.

\paragraph{Realistic Instances}

We construct realistic instances from real-world auction markets in two steps. We first take all bidding data for a country in a one-hour interval and use it to create $n$~buyers and $m$~goods. For that data, we identify the $n$ buyers that participate in the most auctions. (Here, `auctions' refer to the data while `goods' refer to the generated instance.) Each of those buyers maps to a buyer in the generated instance. For now, we define the goods in the instance as the auctions that include at least one of the $n$ buyers. We set the bid in each buyer-auction pair to be the value of the buyer for the corresponding good. The budget for each buyer was set to the buyer's original budget multiplied by a single scalar, calibrated to get the percentage of budget-constrained buyers equal to what was observed in the real-world auction market.

These instances were too large for the MIP to solve. In a second step, we reduce the size through clustering: we apply the $k$-means algorithm to the goods generated so far, using the $n$-dimensional vector of values for that good as features. The goods in the resulting instance were the resulting clusters. Each buyer valuation for a (cluster-level) good is set to the sum of valuations among goods in the cluster. To generate the final buyer budgets, we execute the clustering algorithm for $k=8$ and choose a single factor to scale the budget such that the paced-welfare-maximizing pacing equilibrium had the same fraction of budget-constrained buyers as that of the original auction market. Finally, we run the clustering algorithm again and generate the required $m$ clusters.

For the computational experiments, we preselected $50$ country-hour pairs from $26$ unique countries, and used $n=10$. In each case, we generated scaled-down instances for $8$ to $15$ clusters, for a total of 450 realistic instances.

\subsection{MIP Scalability} \label{sec:mip_scalability}

We start by exploring the size of instances one can solve with the MIP. It is evident that the MIP will not scale to the size of real-world pacing instances, which may involve tens of billions of auctions in a single day. However, less clear is how long it takes to solve instances of different size, and whether some structure in problem instances or MIP objectives were harder to solve than others. The larger the instances that the MIP can solve, the better equipped we are to use the MIP to answer other empirical questions.

The high-level experimental setup is: (1) We generated the 2450 stylized and 450 realistic instances mentioned earlier. (2) We solved each instance using different versions of the MIP. (3) We computed the fraction of instances that were optimally solved, broken down by solution method and instance features. 

We considered several versions of the MIP, which we name within parenthesis: the pure feasibility MIP as defined by \eqref{eqn:budget}-\eqref{eqn:distinctwr} ({\em feasibility}\/), the MIP with the relaxed version of \eqref{eqn:alpha} ({\em relaxed feasibility}), and the feasibility MIP with objectives that minimize or maximize revenue or paced welfare ({\em min revenue, max revenue, min paced welfare, max paced welfare}, respectively). All computations were done with a time limit of 5 minutes using Xpress Optimization Suite~8.0 \cite{xpress} on a server with 24 single-core Intel Haswell CPUs running at 2.5GHz and 60GB of RAM. 
The relatively short timeout allowed us to run the study for our extensive set of instances. (We tested a longer runtime on 98 randomly-selected instances with 9-14 buyers and 9-14 goods. With two hours of runtime we solve 30 of those instances, as opposed to 17 with a five-minute runtime. Thus additional time yields some improvement, but does not greatly extend scalability.) Solutions were programmatically checked to make sure they satisfy the pacing equilibrium conditions.

\paragraph{Results} 

\begin{figure}[tp]
  \centering
	\includegraphics[width=\columnwidth]{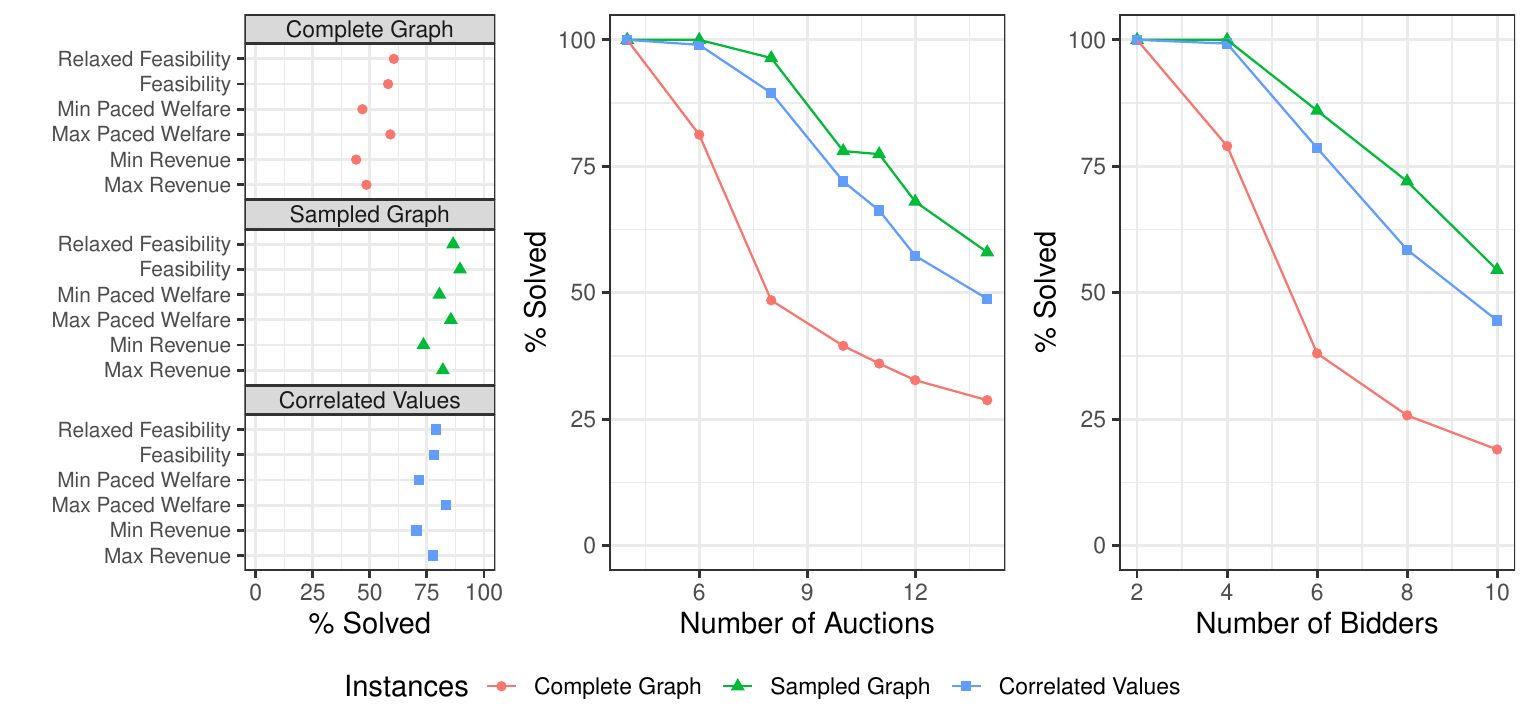}
  \caption{Percentage of stylized instances solved within 5 minutes by objective and graph size}
	\label{fig:ipruntimes}
\end{figure}

We report the percentage of instances for which equilibria were found (as opposed to timing out) and report results by three breakdowns: MIP objective, number of goods $m$, and number of buyers $n$.
Figure~\ref{fig:ipruntimes} shows how runtimes scale on stylized instances for each breakdown.
We observe that complete graph instances are harder to solve than sampled or correlated instances. This is not surprising, since they require more decision variables than the other two types. The instance type played a larger factor in whether an instance was solved than did the objective type. Still, we observe some differences between objective types.
Across all three sets of stylized instance distributions, the value-maximizing MIP objectives were solved more often than value-minimizing objectives.
Paced welfare objectives were solved slightly more often than revenue objectives.
Neither feasibility MIP greatly outperformed the paced-welfare-maximizing MIP.
The feasibility MIP solved about as many instances as the relaxed version.

Grouping the realistic instances by $m=8,\ldots,15$, the MIP respectively found the equilibria for all the sub-cases in 43, 37, 30, 25, 19, 17, 16, and 15 instances out of the 50 combinations of country-hour pairs that we considered. In agreement with the stylized instances, complexity increases as instances comprise more goods.

\subsection{Empirical Analysis of Equilibria}\label{sec:empirical_eqa}

In this section, we use the MIP to improve our understanding of equilibrium properties among instances we can solve. 
We focus on two properties of pacing equilibria that we covered in our analytical work, and explore how they change as a function of instance size.
First, do we frequently observe large \emph{empirical differences in equilibria}---that is, between the value-maximizing and value-minimizing equilibria---or do such gaps only arise in pathological examples? Second, how frequently do we observe large incentives for advertisers to \emph{misreport bids and budgets}? How are such incentives affected by features of the pacing instance?

\subsubsection{Empirical Differences in Equilibria}

Even though we have seen large gaps for instances constructed carefully for that purpose, we see that this is not the case for the practically-inspired instances that we put forward.
For each of those instances, we solve a pair of MIPs to find the difference between the value-maximizing and value-minimizing equilibrium. We then measure such gaps across sets of instances.
We measure equilibrium gaps for revenue, welfare, and paced welfare. 
For revenue and paced welfare, we have a MIP that finds extremal equilibria.
Measuring welfare gaps is less straightforward: Since we could not represent the welfare objective as a linear expression, we do not have a MIP to optimize. Instead, we compute the gap among the MIP solutions optimizing the other objective functions.
Hence, the reported gaps in social welfare are a lower bound on the maximal achievable gap. 
We report the following metrics:
\begin{description}
  \item[\emph{Pairs \%}:] the percentage of instances for which a pair of MIPs (objective-maximizing and -minimizing) both returned prior to the five-minute timeout. (For the welfare objective, a pair was counted if any two MIPs where any objectives were solved.)
\item[\emph{No Gap \%}:] the percentage of paired instances with no gap in objective value.
\item[\emph{Max Gap}:]the largest observed gap across instances, as a percentage of the objective-maximizing value.
\end{description}

\begin{table}[t]
\caption{Gaps between equilibria for different objectives and classes of problem instances.}
\label{tab:gap_results}
\centering\scriptsize
\begin{tabular}{llrrr}\hline
Objective & Instances & Pairs \% & No Gap \% & Max Gap \\\hline
Revenue & Complete Graph & 40.6 & 97.2 & 44.2 \\ 
  Revenue & Sampled Graph & 73.1 & 96.9 & 33.8 \\ 
  Revenue & Correlated Values & 68.4 & 99.7 & 13.3 \\ 
  Paced Welfare & Complete Graph & 46.3 & 91.4 & 39.7 \\ 
  Paced Welfare & Sampled Graph & 80.0 & 93.6 & 16.1 \\ 
  Paced Welfare & Correlated Values & 70.6 & 93.7 & 42.9 \\ 
  Welfare & Complete Graph & 60.0 & 92.4 & 5.3 \\ 
  Welfare & Sampled Graph & 88.0 & 92.2 & 2.6 \\ 
  Welfare & Correlated Values & 81.7 & 94.2 & 2.9 \\\hline
\end{tabular}
\end{table}

\paragraph{Results} 
Computing pacing equilibria for the stylized instances, we find that  
for the majority of instances, there was no gap in the objective value across equilibria. In some cases, however, gaps were as large as 44.2\%. 
The welfare objective had smaller gaps: all were less than 5.3\%, but recall that for this objective those are only lower bounds.
Table~\ref{tab:gap_results} summarizes the empirical study of gaps.

Some instances remained unsolved at the end of the time limit,
and it could be that these instances also have large gaps.
To test this hypothesis, we looked at runtimes of solved
instances with large gaps. The conclusion is that these tend to be instances
that can be solved quickly. For instance, none of the instances with nonzero gaps took more than
ten seconds to solve. 
This analysis suggests that longer runtimes arise from sparseness of equilibria,
rather than from their multiplicity.

The conclusions on equilibrium gaps also extend to realistic instances. We saw essentially no differences for the $43$ realistic instances for which the scaled-down instance had solutions for all objectives for at least one number of clusters $m$. Only one unclustered instance had any revenue difference (and only for a single choice of $m$): a difference of $0.03\%$. Only two unclustered instances had welfare differences: one with a difference of $0.03\%$ in a single clustering, and one with about $2.9\%$ in every clustering. Only one unclustered instance had paced welfare differences: $5.7\%$ in the worst clustering, but $2$-$3\%$ in the remaining clusterings.

Overall, these results are promising:
Although our theoretical results demonstrated that gaps can be large in theory, 
we found empirically that most gaps on instances we considered were small, and
often times, there was no gap at all. More summary statistics on gaps are given in Appendix~\ref{sec:experiment:appendix} in the e-companion.

Even when there is no gap in revenue or welfare across computed equilibria, multiple equilibria with the same objective value could nevertheless exist. 
To test whether this occurs, we looked at 156 instances where at least two MIPs finished within five minutes. Amongst those instances, we identified the subset of instances for which two or more MIP models returned non-identical vectors of pacing multipliers. We found that 144 instances (i.e., all but 12 instances) returned the same vector of pacing multipliers across all MIP models that returned a solution within five minutes. For the 12 instances that had multiple equilibria, either revenue or paced welfare differed as well. While this does not prove that there are no other solutions with the same objective value, it does suggest that this does not happen often. We expect that the smoother the instances are, the less likely it is that the instance will admit multiple equilibria with the same objective values.

\subsubsection{Robustness to Misreporting}
\label{sec:misreporting}

We now explore whether buyers can improve their performance by misreporting bids or budgets 
to affect the pacing equilibrium in the proxy-bidder setting where the platform adjusts pacing multipliers.
Although Proposition~\ref{prop:pacing_optimality} implies that buyers do not benefit from 
misreporting given the current pacing equilibrium, buyers may improve their
utility if they can influence the resulting equilibrium. 
Example~\ref{ex:misreporting} highlighted such a case where a buyer significantly
increased its utility by misreporting; we investigate the extent to which 
misreporting may be a problem in practice with a computational study.
Concretely, we study the relation between the incentive to misreport and market thickness.

\begin{figure}[t]
  \centering
	\includegraphics[width=1.0\columnwidth]{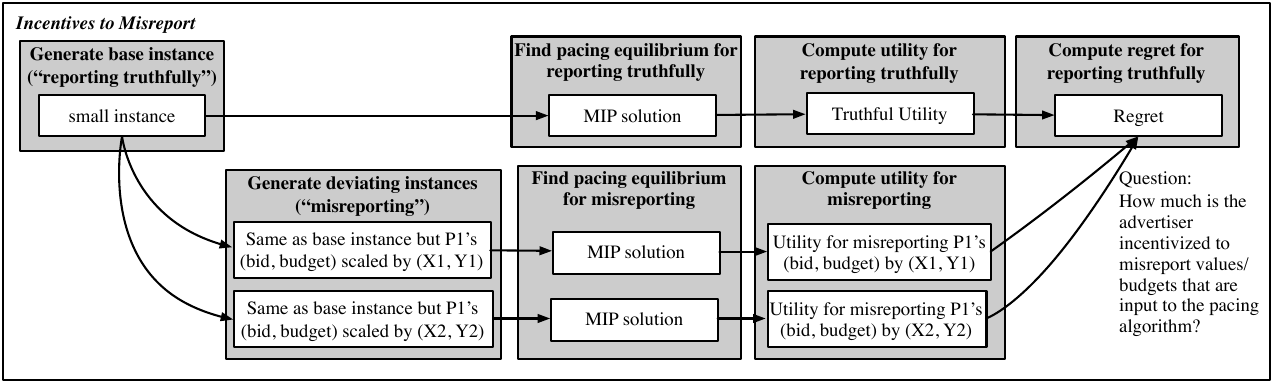}
	\caption{An overview of the steps to measure an advertiser's gain in utility for strategically misreporting values and budgets to the proxy bidder}
	\label{fig:setup:manip}
\end{figure}

We consider two possibilities in how a focal buyer may strategically misreport.
Motivated by pacing multipliers, a simple case is when the buyer misreports values by choosing bids that arise from multiplying all values by a fixed constant, and submitting some alternative budget. We refer to this case as uniform deviations. A second case is when the focal buyer performs a grid search to find optimal bids for each good. For this latter case we omit budget deviations, since the number of MIPs that need to be solved during the grid search is already substantial.
We performed both studies and show the output in Figure~\ref{fig:manipulableGrid}.
Results turn out to be of similar magnitude but marginally higher for grid search.

Since the case of uniform deviations is closely aligned with the setting of our model, we describe it first.
Here, buyers submit a budget $B_i$ and a single valuation $v_i$ for
a generic good. The valuation for specific goods is then set to $v_{ij}= v_i
\gamma_{ij}$, where we assume that $\gamma_{ij}$ is fixed ahead of time. This setting models
paced internet auction markets, where buyers typically submit their
budget, targeting criteria, and their valuations for clicks.
Their valuation for an individual impression is then typically calculated as their valuation for a click times the click-through rate of the impression-ad pair. Because the click-through rate $\gamma_{ij}$ is estimated by the platform, each advertiser can only affect their overall bid for a click.

\begin{figure}
	\centering
  $~$
  \hfill
  \includegraphics[width=0.45\columnwidth]{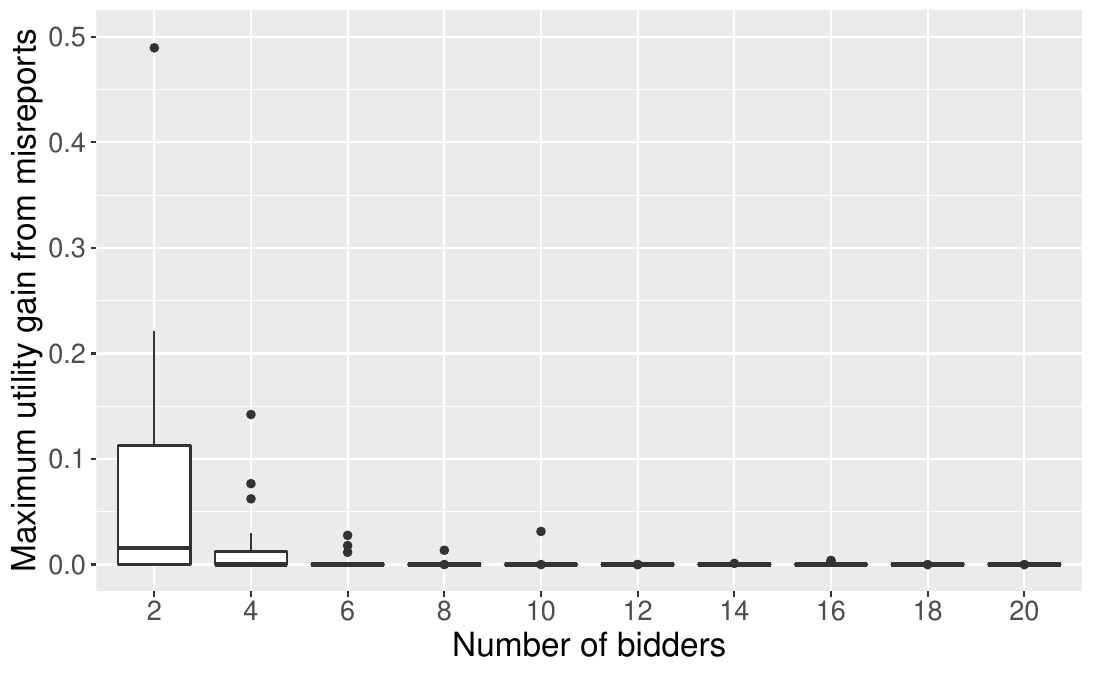}
  \hfill
  \includegraphics[width=0.45\columnwidth]{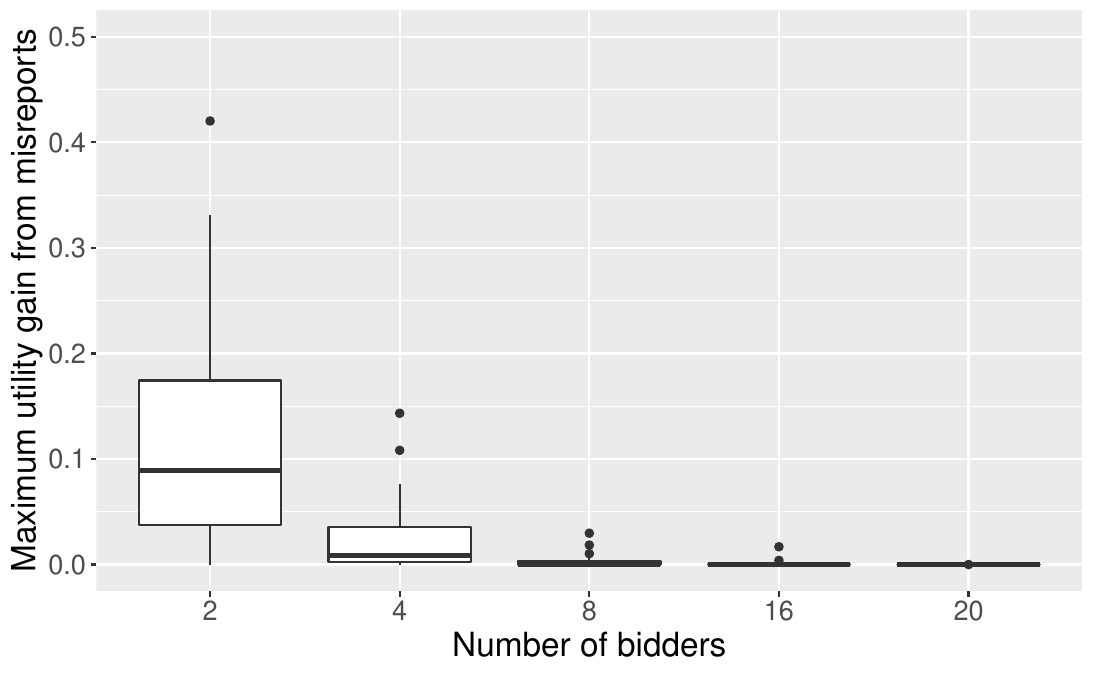}
  \hfill
  $~$
  \caption{
    Utility gains when misreporting. 
    The utility gain is defined as $v^* - v$, where $v^*$ is the optimal value when misreporting and $v$ is the truthful value.
  {\em Left}: Advertiser uniformly manipulates all bids by multiplying them with a constant. 
  {\em Right}: Advertiser explores all bid manipulations using grid search.
   } 
	\label{fig:manipulableGrid}
\end{figure}

We generated $40$ instances, each with $4$ auctions, $2$ through $20$ buyers, and budgets and
valuations generated according to the complete-graph setting described in Section~\ref{sec:problem_instances}.
To model misreporting, we assume that buyers may multiply their budgets and valuations, respectively, by scalars $(\beta_i, \nu_i) \in \{0.6,0.8,\ldots,1.4\} \times \{0.5,0.6,\ldots,1.4\}$.
We consider that only a focal buyer per auction is strategically misreporting while the others remain true to their actual values. 
The focal buyer submits scaled budgets and valuations to the pacing mechanism, instead of the actual ones. With this setup, we solve the MIPs and compare the resulting utility received by the focal buyer across all $(\beta_i, \nu_i)$.
An overview of the setup is shown in Figure~\ref{fig:setup:manip}.

\paragraph{Results} 

We say that the focal buyer in an instance has an incentive to misreport when there is a combination $(\beta_i, \nu_i)$ for which the buyer's utility is higher than that when true values are reported.
The left panel of Figure~\ref{fig:manipulableGrid} shows the utility gain that results from applying the multiplicative scalars.
With more buyers, the incentive to misreport disappears rapidly. 
Since we only tried a discrete and finite set of scalars, one could argue that
perhaps it would have been optimal to use another scalar that was not in the set.
To provide evidence that this is unlikely, we tried a much finer
discretization for 10 instances composed of 6 to 14 buyers. This
time we considered scalars $(\beta_i, \nu_i) \in \{0.6,0.65,\ldots,1.4\} \times
\{0.5,0.55,\ldots,1.4\}$; i.e., a step size of $0.05$ instead of $0.1$. Out of these 50 instances,
we found only one instance for which a deviation increased the utility. It is also unlikely 
that we needed a bigger interval of scalars: none of the
instances we tested resulted in the optimal scalar being at the extremes of an interval.

The right panel of Figure~\ref{fig:manipulableGrid} shows the utility gain that results from performing grid search across each valuation.
Here, we consider the more general case where the buyer can submit an independent valuation $\nu_{ij}v_{ij}$ for appropriately chosen scalars $\nu_{ij}$. We considered a coarser discretization for the values of $\nu$ and for how many buyers to consider than before for computational reasons: we consider every misreport $\left\{ \nu_{ij}  \right\}_{j=1}^4 \in [.5, .9, 1., 1.1, 1.5]^4$. We also only consider 2, 4, 8, 16, and 20 buyers, and generate 20 instances per size.
As can be seen in the figure, although the maximum deviation is marginally higher, it is the same order of magnitude and it rapidly goes to zero as the market gets denser, thus not affecting the conclusions drawn in the first part.

In summary, we found that the buyer was rarely incentivized to manipulate its valuation or budget, and that the incentive decreased as the size of the instance grew. Collectively, the empirical conclusions on equilibrium properties reassure us that some of these potential problems may not actually be so in practice.

\subsection{Seeding Dynamic Instances}
\label{sec:improve_heuristics}

Real-world pacing heuristics rely on tractable adaptive algorithms that update buyers' pacing multipliers over time. Recent results have shown that such algorithms can converge to stable pacing multipliers in the limit, given simplifying distributional assumptions~\citep{balseiro2017dynamic}. While convergence in the limit is a positive result, the \emph{rate} at which the learning algorithm converges is important in practice: The longer the pacing algorithm takes to converge, the worse it is at optimizing the buyer's utility. This motivates us to study how the stability of adaptive algorithms can be improved.

The adaptive algorithm we use for this set of experiments is from \citet{balseiro2017dynamic}, which we refer to as \AdaptiveAlg; see Algorithm~\ref{alg:adaptive}.
\AdaptiveAlg takes as input a pacing instance $\Gamma$, a vector of initial pacing multipliers $(\alpha_i^\texttt{init})_{i \in N}$, a minimum allowable pacing multiplier $\alpha^\texttt{min}$, and a step size $\epsilon$, which affects how much the multiplier changes across auctions.
After each auction $j \in M$, each buyer $i$ updates its multiplier based on the difference between the buyer's spend and its \emph{target per-auction expenditure}, which is the average amount to spend per auction to perfectly exhaust the budget. In this paper, Algorithm~\ref{alg:adaptive} uses the notion of a multiplier $\alpha$, which differs from the notion of a multiplier $\mu$ in~\citet{balseiro2017dynamic}; the relationship is $\alpha = 1 / (1 + \mu)$. Other minor differences from \citet{balseiro2017dynamic} are that (1) we made the initial multipliers an explicit parameter to the algorithm; and (2) we removed per-buyer subscripts for $\alpha^\texttt{min}$ and $\epsilon$, since all buyers use the same value in our experiments.

\begin{algorithm}[]\label{alg:greedyVCG}
\KwIn{Pacing instance $\Gamma = (N, M, (v_{ij})_{i \in N, j \in M}, (B_i)_{i \in N})$; initial multipliers $(\alpha_i^{\texttt{init}})_{i \in N}$; minimum multiplier $\alpha^{\texttt{min}}$; step size $\epsilon$.}
	
	\For{$i \in N$} {
	Set target expenditure $\rho_i = B_i / m$\;
	Initialize remaining budget $B_{i1} = B_i$ and multiplier $\alpha_{i1} = \alpha_i^\texttt{init}$\;
	}	
		\For{$j \in M$} {
	
		Each buyer $i$ places bid $b_{ij} = \min(v_{ij} \alpha_{ij}, B_{ij})$.
		
    The auction outputs an allocation $(x_{ij})_{i \in N}$ and payments $(s_{ij})_{i \in N}$.
		
		Each buyer $i$ updates its multiplier
		$\alpha_{i,j+1} = \max(\alpha^\texttt{min}, 1 / \max(1, 1 / \alpha_{ij} - \epsilon (\rho_i - s_{ij})))$ and remaining budget
		$B_{i,j+1} = B_{ij} - s_{ij}$.
	}
	\Return bids, allocations, and payments\;
\caption{\AdaptiveAlg~\citep{balseiro2017dynamic}}
\label{alg:adaptive}
\end{algorithm}

We now describe the computational study that evaluates using the MIP to improve the adaptive algorithm.
Our study consists of running the adaptive algorithm on \emph{large instances}, seeded from the pacing multipliers of a \emph{small instance} that compactly represents each of them.
We compare outputs under two types of initial multipliers: 
a unique constant for all buyers which is what one would do without additional information (e.g., each buyer starts with multiplier~0.5),
and the pacing multipliers returned by the MIP for the original (static) instance.
For each set of initial pacing multipliers, we determine parameters $\epsilon$ and $\alpha^\texttt{min}$
through grid search by choosing those that minimize the
\emph{average ex-post relative regret} (i.e., the average amount that a buyer could
have improved its utility by playing a single best-response multiplier, given
the other bids are fixed).
This is summarized in Figure~\ref{fig:brperf_exp_short}.
(See Appendix~\ref{sec:brdynamics} in the e-companion for analogous experiments on using the MIP to improve existing algorithms of a one-shot model, where the tractable algorithm in that case is best-response dynamics.)

\begin{figure}[t]
  \centering
	\includegraphics[]{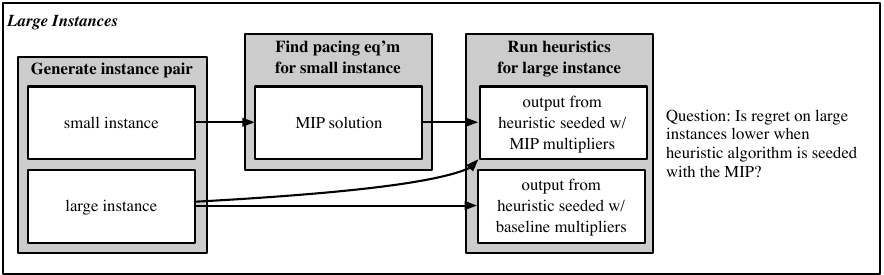}
  \caption{Steps to measure the benefit from warm-starting heuristics with the MIP}
	\label{fig:brperf_exp_short}
\end{figure}

In the case of stylized instances, we create $K$ complete-graph instances, and scale each one up 
to produce a large instance. To scale, we make $C$ copies of each
good, and scale budgets by $C$. For each edge in the scaled-up
graph, we perturb the buyer-good value by adding Gaussian noise with mean $0$
and standard deviation $\sigma$ (this noise parameter $\sigma$ is different from that of the correlated stylized instances).
The parameters we used for these experiments are $K=6$,
$C=500$, $\sigma \in \{0, 0.1, 0.5\}$, $\alpha^\texttt{min} \in \{0.1, 0.05\}$,
and $\epsilon \in \{0.01, 1, 2\}$.
More details on these instances are given in Appendix~\ref{sec:experiment:appendix} in the e-companion.

In the case of realistic instances, as explained in Section~\ref{sec:problem_instances}, we take a large instance constructed from real data, and we cluster it to make a compact representation.
In this case, we solve the MIP corresponding to the clustered instance and then provide the resulting pacing multipliers as input when running \AdaptiveAlg on the original one without clusters.

\paragraph{Results}

As shown in Figure~\ref{fig:real_exp}, running \AdaptiveAlg with MIP-based initial multipliers produces a lower regret than with other choices of initial multipliers. 
Performance of the MIP-based solution degrades as the noise parameter $\sigma$ grows, but even at the highest levels we considered, the MIP-based solution outperformed the baseline solutions. 
When using fixed initial multipliers, the resulting regret is highly sensitive to choices in the step size: low initial multipliers would often not reach the MIP's equilibrium multipliers by the time the algorithm terminated.
For realistic instances, Figure~\ref{fig:real_exp} also shows that the regret experienced by buyers when starting from the MIP-based initial multipliers was lower than in the other cases, for every learning rate $\epsilon$ we considered.
More strongly, the worst learning rate for the MIP was better than the best learning rate for any of the baselines. These findings were robust to different number of clusters~$m$ when producing the realistic instances. Surprisingly, $8$ clusters was enough to find good multipliers and increasing $m$ to $15$ clusters did not reduce the regret.

These experiments for the dynamic setting leave us optimistic about the potential value of computing static pacing multipliers with the MIP. Using the MIP to warm-start an adaptive algorithm on these larger instances resulted in better convergence, and these improvements were robust to noise in the MIP input. 
Such robustness is important for two reasons: First, it suggests that the MIP does not need the exact valuation distribution to be useful (which is unlikely to be known in practice); second, it suggests that the valuation distribution could be compressed to create a smaller (approximate) problem instance that could be tractably solved by the MIP. In follow up work, \citet{kroerApproxMarketEq} address both of these issues.

\begin{figure}
  \centering
	\includegraphics[width=0.47\columnwidth]{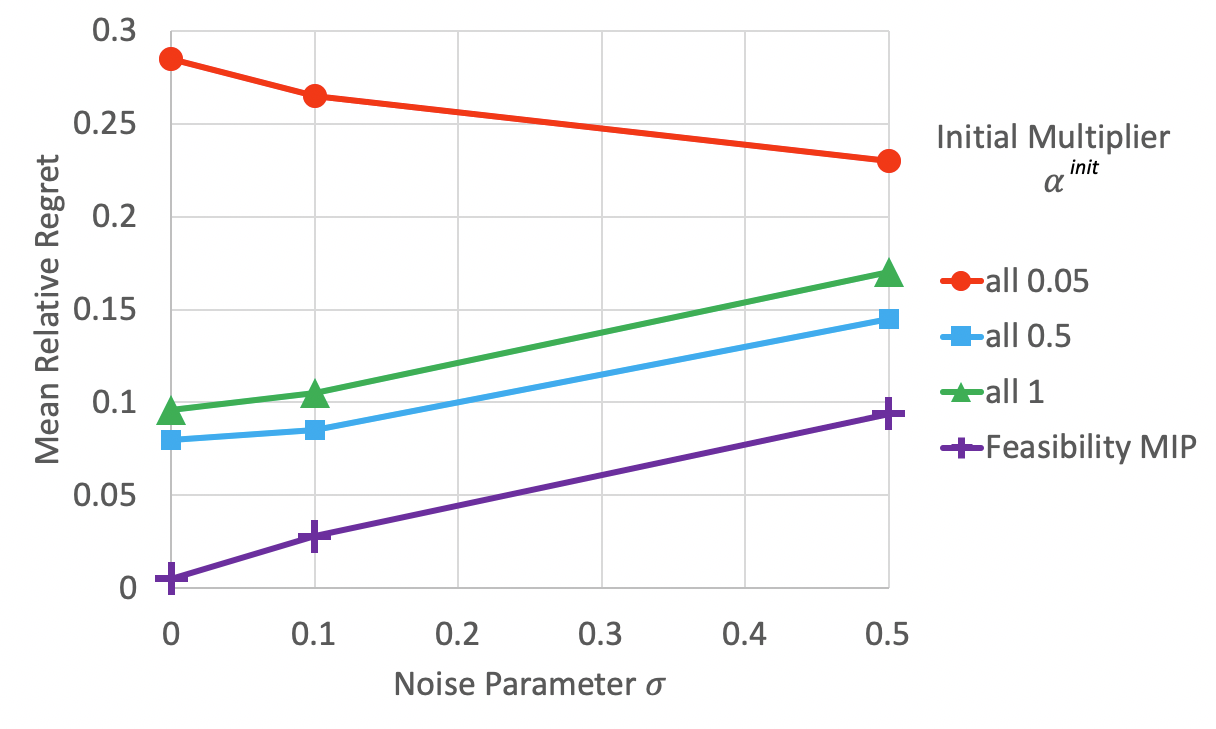}
  $~$
  \includegraphics[width=0.47\columnwidth]{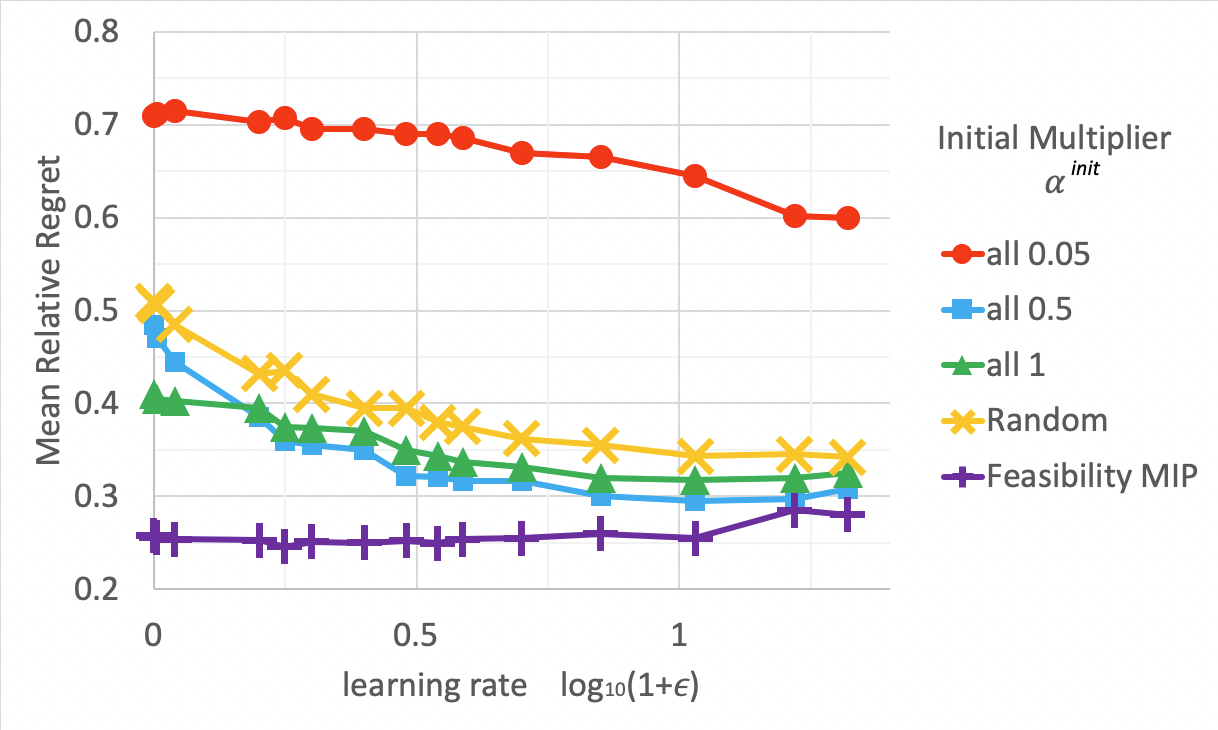}
	\caption{
  Mean relative regret from running \AdaptiveAlg.
  Each curve plots different initial pacing multipliers $\alpha_i^{init}$.
  {\em Left}: Stylized instances with random perturbations. Regret as a function of the noise parameter $\sigma$. 
  {\em Right}: Realistic instances with 8 clusters (no noise). Regret as a function of the learning rate~$\epsilon$ (shown in log scale as $\log10(1+\epsilon)$).
  }
  \label{fig:real_exp}
\end{figure}

\subsubsection*{Interpretation of MIP Solution}

A natural question is whether the MIP output when solving static pacing has any interpretation regarding the dynamic pacing setup we discussed. 
We provide evidence that the output of the adaptive algorithm on the large instances matches the pacing multipliers computed by the MIP on the small instances.
Appendix~\ref{se:dynamics} in the e-companion describes a limit dynamics model and prove that a solution in that model is stable if and only if it constitutes a pacing equilibrium (which the MIP outputs). 

In a similar setup to the study of dynamic pacing with stylized instances, we randomly sampled $K$ complete-graph problem instances
and run \AdaptiveAlg with initial multipliers equal to the feasibility MIP output. For each instance, we compute the absolute difference between the MIP fractional allocation and the \emph{empirical allocation} (that is, the fraction of goods won in the scaled-up instance that corresponded to the same good in the original instance). The parameters we used for this experiment were $K = 20$, $C = 50$, $\alpha^\texttt{min} = 0.05$, and $\epsilon = 10^{-4}$; other parameter settings gave similar results.

\paragraph{Results} 

Figure~\ref{fig:dancing} (Left) shows a summary of the absolute differences between the fractional allocations output by the MIP and \AdaptiveAlg. 
The difference between the fractions allocated never exceeded 0.07, and over 75\% of the time, the difference was less than 0.02. 
To understand why the empirical allocations so closely match the MIP allocations, see 
Figure~\ref{fig:dancing} (Right) for an illustrative example. The figure shows the per-auction bids for a particular pacing instance, good type, and subset of buyers. In the original version of this instance, the feasibility MIP found a solution in which three buyers won a fractional allocation of the good. When we started \AdaptiveAlg from the MIP output multipliers, the induced bids danced around the winning price such that the empirical allocation for these buyers nearly matched the MIP output (with allocation values of (0.51, 0.24, 0.25) versus (0.52, 0.24, 0.22) for each respective buyer).

These results illustrate that the MIP fractional multipliers have a meaningful interpretation for larger instances in which one runs an adaptive algorithm. 

\begin{figure}
	\centering
	$~$
	\hfill
	\includegraphics[scale=0.9]{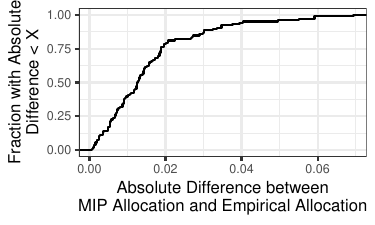}
	\hfill
	\hfill
	\hfill
	\hfill
	\includegraphics[width=0.45\columnwidth]{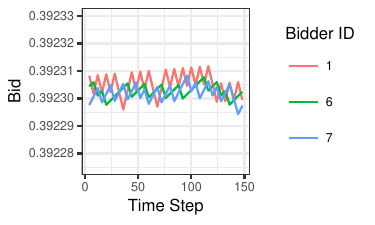}
	\hfill
	$~$
  \caption{{\em Left}: An empirical CDF over absolute differences between the empirical allocation and the MIP allocation. Empirical allocations were approximately equal to the MIP allocation across all instances. {\em Right}: An example of buyers adaptively adjusting their pacing multipliers to effectively win a fraction of the good type; the empirical allocation in this instance was approximately equal to the feasibility MIP fractional allocation.}
	\label{fig:dancing}
\end{figure}

\section{Conclusion}\label{sec:conclusion}

In auction markets, buyers with budgets are not necessarily best off submitting
their true valuations.
We considered {\em multiplicative pacing} and proved its optimality from the
buyer's viewpoint (Proposition~\ref{prop:pacing_optimality}). We introduced a
notion of {\em pacing equilibrium} (Definition~\ref{def:pacing} and
Proposition~\ref{the:best_responding}), proved (a) their existence in
Theorem~\ref{th:existence}, (b) close relations to competitive equilibria in
Section~\ref{sec:competitive_equilibrium}, and that (c) finding equilibria
maximizing welfare and revenue is NP-hard in Theorem~\ref{th:max_rev}. We gave a
MIP formulation for finding optimal pacing equilibria and evaluated it
experimentally. We found that although multiple equilibria may exist, their
paced welfare and revenue are frequently similar.
For adaptive pacing, we found that regret-based dynamics arrived at allocations
near our MIP-based solutions, and that these allocations were improved by
warm-starting with solutions from our MIP, even when the MIP input was noisy.
Our experimental findings were robust to several different random models of
markets, as well as markets generated from real-world auction data.

While the MIP can only be run on small-enough instances, its solution has an interpretation for larger instances, both when doing clustering or when valuations were drawn jointly across buyers in proportion to valuations from the original instance. 
Using the MIP to warm-start an adaptive algorithm on these larger instances indeed results in better convergence, and these improvements were robust to noise or to different ways to generate the smaller instance from the larger instances.

A few open questions remain: What is the computational complexity of
finding an {\em arbitrary} pacing equilibrium? Can
we generalize to multiple-slot auctions or to a
dynamic setting with uncertainty about future auctions?
Can we make further realistic assumptions on the primitives to get tractability or stronger results?
In dynamic settings, 
how do we improve the convergence to optimal equilibria? 
One direction is to explore how to best compress extremely large problem instances---those with many buyers and many more auctions---so that the MIP provides valuable output for warm-starting large-scale dynamic pacing problems.

%The proof will be given in the e-companion to this paper.

% Acknowledgments here
\ACKNOWLEDGMENT{%
%\section*{Acknowledgments}

We would like to thank 
Santiago Balseiro, 
Ozan Candogan, 
Roberto Cominetti,
Nikhil Devanur,
Yoni Gur, 
Renato Paes Leme,
Debmalya Panigrahi,
Marco Scarsini,
Okke Schrijvers,
Marc Schroder,
and 
Chris Wilkens
for very useful discussions on the technical content of this paper.
In addition, we extend our sincere thanks to 
the review team of Operations Research,
conference PC members and reviewers at WINE, 
seminar participants at Duke CS-Econ, Melbourne Business School, Michigan Ross, MIT ORC and Stanford GSB, 
and attendees of 
the AAMAS-IJCAI workshop on Agents and Incentives in Artificial Intelligence (AI$\hat\ $3),
the AI and Marketing Science workshop at AAAI,
the Algorithmic Game Theory and Data Science workshop at EC,
the Barbados AGT workshop,
the International Conference on Game Theory,
and 
the Schloss Dagstuhl seminar on Traffic Models.
Their remarks and questions allowed us to significantly improve the manuscript and its presentation.
Finally, we also thank the review team in Operations Research for their insightful comments and questions that also helped us improve this article.

}

% References here (outcomment the appropriate case)

% CASE 1: BiBTeX used to constantly update the references
%   (while the paper is being written).
%\bibliographystyle{informs2014} % outcomment this and next line in Case 1
%\bibliography{<your bib file(s)>} % if more than one, comma separated

% CASE 2: BiBTeX used to generate mypaper.bbl (to be further fine tuned)
%\input{mypaper.bbl} % outcomment this line in Case 2

%If you don't use BiBTex, you can manually itemize references as shown below.

\bibliographystyle{plainnat}
\bibliography{refs}

%% Here starts the e-companion (EC)
%%%%%%%%%%%%%%%%%%%%%%%%%%%%%%%%%%%%%%%%%%%%%%%%%%%%%%%%%%
\ECSwitch

%\ECDisclaimer
%%%%%%%%%%%%%%%%%%%%%%%%%%%%%%%%%%%%%%%%%%%%%%%%%%%%%%%%%%

%%% Main head for the e-companion
% A general heading for the whole e-companion should be provided here as in the example above this paragraph.
\ECHead{Multiplicative Pacing Equilibria in Auction Markets: E-Companion with Proofs and Additional Material}

\section{Missing Proofs}\label{sec:missingproofs}

\subsection{Pareto Optimality of Competitive Equilibria}

We prove formally that competitive equilibria are Pareto optimal, which is a known result, but we do it specialized to our setting of buyers with budgets for concreteness.

\bigskip
\noindent{\bf Proposition~\ref{prop:ce_pareto}}
  (The first fundamental theorem of welfare economics).
Any competitive equilibrium is Pareto optimal (when considering the utilities of all buyers and the seller).
\begin{proof}{Proof.}
  For the sake of contradiction, suppose there is a Pareto dominating allocation
  of goods and money. Let $p_j$ denote the price of good $j$ in the competitive
  equilibrium, $S_i$ the bundle received by buyer $i$ in the competitive
  equilibrium (and $S'_i$ in the dominating allocation), and $t_i = p(S_i) \leq
  B_i$ the amount of money buyer $i$ spends in the competitive equilibrium (and
  $t'_i \leq B_i$ in the dominating allocation). Here, we let $p(S)=\sum_{j \in
    S} p_j$ be the price of any bundle of goods $S$ under the competitive
  equilibrium prices. Now, for every buyer $i$, it must be the case that $t'_i
  \leq p(S'_i)$. This is because either $p(S'_i)>B_i$ or $v_i(S_i)-p(S_i) \geq
  v_i(S'_i)-p(S'_i)$ by the property of competitive equilibria, yet $v_i(S_i) -
  p(S_i) \leq v_i(S'_i) - t'_i$ by Pareto dominance. If $t'_i < p(S'_i)$ for
  some $i$, then it follows that the total payment in the dominating allocation
  is less than the sum of the good prices, contradicting that the seller is at
  least as well off. On the other hand, if $t'_i = p(S'_i)$ for all $i$, then no
  buyer is better off, and also the seller is just as well off, again
  contradicting Pareto dominance.
\hfill\Halmos
\end{proof}

\subsection{NP-Hardness of Computing Pacing Equilibria}

We first note the following proposition about our Example~\ref{ex:binary_choice}
for modeling binary choices.

\medskip

\begin{proposition}
  In Example~\ref{ex:binary_choice}, when $\alpha+\delta < 1/3$, no
  equilibrium satisfies $\min(\alpha_1,\alpha_2) \geq 3 \alpha$.
\label{prop:binary_choice}
\end{proposition}
\begin{proof}{Proof.}
  The reason is that if such an equilibrium existed,
  the total price of the first two goods would be at least $6\alpha K_2 = 3
  (1-\alpha-\delta)K_1 > 2 K_1$ (which follows from the statement).
  This is the combined budget of the two buyers,
  resulting in a contradiction.
\hfill\Halmos
\end{proof}

With this proposition we are ready to prove our complexity result.

\bigskip
  \noindent{\bf Theorem~\ref{th:max_rev}.}
  {\sc MAX-REVENUE-PACING}, {\sc MAX-WELFARE-PACING} and {\sc MAX-PACED-WELFARE-PACING} are NP-complete.
\begin{proof}{Proof.}
  We reduce an arbitrary {\sc 3SAT} instance to the following {\sc
    MAX-REVENUE} instance. We set $T$ equal to the number of clauses, plus $4$
  times the number of variables, in the {\sc 3SAT} instance. For every
  variable $x_j$, we create a copy of Example~\ref{ex:binary_choice}, consisting
  of buyers $1^{x_j}, 2^{x_j}$ and goods $1^{x_j}, 2^{x_j},3^{x_j},4^{x_j}$,
  with bids as specified in the example, using $K_1=4$, $\alpha=1/4$,
  $\delta=0$, and (hence) $K_2=6$. Each of these goods will only be bid on by
  the buyers corresponding to its own variable (the other buyers have
  valuation $0$ for them). However, the buyers will bid on other goods as well,
  namely goods corresponding to the clauses. Specifically, we associate buyer
  $1^{x_j}$ with the literal $+x_j$, and buyer $2^{x_j}$ with the literal
  $-x_j$. A buyer values a clause good at $1$ if its literal occurs in that
  clause, and at $0$ otherwise. Finally, we add a single buyer with unlimited
  budget that values every clause good at $2$. Hence, this buyer will
  necessarily win all the clause goods, at price at most $1$ each.

  Suppose a satisfying assignment exists. If $x_j$ is set to {\em true}, set
  $\alpha_{1^{x_j}}=1$ and $\alpha_{2^{x_j}}=\alpha$; otherwise, set
  $\alpha_{1^{x_j}}=\alpha$ and $\alpha_{2^{x_j}}=1$. This depletes the budgets
  of the buyers corresponding to variables, resulting in a revenue of $4$ times
  the number of variables. Moreover, for every clause good, the unlimited-budget
  buyer faces one of the variable buyers with a multiplier of $1$, since we
  had a satisfying assignment. Hence this buyer pays an amount equal to the
  number of clauses. Hence the {\sc MAX-REVENUE-PACING} instance has a solution.

  Conversely, suppose the {\sc MAX-REVENUE-PACING} instance has a solution.
  Then, the unlimited-budget buyer must pay at least an amount equal to the
  number of clauses. Because she pays at most $1$ on each clause good, it
  follows that she must pay exactly $1$ on each clause good. Hence, at least one
  of the buyers corresponding to positive literals in each clause must have a
  multiplier $1$. But since, by Proposition~\ref{prop:binary_choice}, at most
  one of the two buyers corresponding to a variable can have a multiplier of
  $1$, it follows that these buyers correspond to a satisfying assignment.

  Now we switch to the welfare objective.
  We reduce an arbitrary {\sc 3SAT} instance to the following {\sc
    MAX-WELFARE} instance. We set up buyers corresponding to variables as in
  the {\sc MAX-REVENUE} proof. We set $\alpha=\delta=\frac{1}{8}$, $K_1=1$, and
  thus $K_2=3$. We let $V,C$ be the sets of variables and clauses in the {\sc
    3SAT} instance, respectively. We set $T$ equal to
\[
\delta K_1 + |V|\left(2K_2+\left(\frac{K_1}{\alpha}+\epsilon\right)\right)=\frac{1}{8} + |V|(14 + \epsilon),
\]
For clauses, a buyer values a clause at value $\frac{\delta K_1}{|C|}$ if its
literal occurs in that clause, and at $0$ otherwise. Finally, we add a single
buyer with unlimited budget that values every clause good at $\frac{\delta
  K_1}{2|C|}$.

Suppose a satisfying assignment exists. Perform the assignment as in the {\sc
  MAX-REVENUE} setting. That gives a social welfare of
$|V|(2K_2+(\frac{K_1}{\alpha}+\epsilon))$ from the variable goods. Furthermore,
for each clause, at least one satisfied-literal buyer has its pacing multiplier
set to $1$, thus winning the clause good, yielding utility $\frac{\delta
  K_1}{|C|}$. Summing over the clauses gives the desired social welfare. Each
buyer can at most win all the clauses, and thus their spend is bounded by
$(1-\delta)K_1+ \delta K_1$, satisfying their budget constraint.

Conversely, suppose the {\sc MAX-WELFARE-PACING} instance has a solution. Then
each clause good must be allocated to a satisfied-literal buyer. But, in order
to beat the unlimited-budget buyer, the satisfied-literal buyer must have a
pacing multiplier of at least $\frac{1}{2}$. By
Proposition~\ref{prop:binary_choice}, this means that the buyer corresponding
to the opposite literal must have a multiplier less than or equal to
$\frac{3}{8}$. Therefore, the buyers with pacing multipliers of at least
$\frac{1}{2}$ correspond to a satisfying assignment.

We can perform almost the same reduction for {\sc MAX-PACED-WELFARE-PACING}. We construct the same set of buyers and valuations. We set $T$ equal to
\[
\delta K_1 + |V|\left(2K_2+ \alpha\left(\frac{K_1}{\alpha}+\epsilon\right)\right)=\frac{1}{8} + |V|\left(7 + \frac{ \epsilon }{8}\right),
\]
If a satisfying assignment exists, we can set the same pacing assignment as
before. The only difference from the previous construction is that the paced
welfare from the variable goods is now $|V|(2K_2+(K_1+\frac{ \epsilon }{8}))$.
Combined with the clause good assignment, this gives exactly the desired paced
welfare.

The converse case becomes simpler. For any {\sc MAX-PACED-WELFARE-PACING}
solution, it must be the case that each variable has at least one buyer with a
pacing multiplier of $1$. To obtain the remaining paced welfare of $\delta K_1$,
these buyers with pacing multiplier $1$ must correspond to a satisfying
assignment.
\hfill\Halmos
\end{proof}

\subsection{Correctness of MIP Formulation}\label{sec:appendix_proofMIP}

\noindent{\bf Proposition~\ref{prop:correct_formulation}.}
    A solution to the MIP 
  \eqref{eqn:budget}-\eqref{eqn:distinctwr} is feasible if and only if it corresponds to the conditions of a pacing
  equilibrium.

\bigskip
\begin{proof}{Proof.}
 Assume that all goods $j$ have some buyer $i$ such that $v_{ij}>0$.
 Otherwise, we preprocess the problem by removing all goods that no buyers are
 interested in.

 First, let $\alpha_i,x_{ij} \in [0,1], s_{ij}\in \Rplus$ be
 a pacing equilibrium for a pacing game. Let all MIP variables be set according
 to their definition as it pertains to the pacing equilibrium. Set $x_{ij}=1$
 if $x_{ij}>0$. If there are multiple buyers with $x_{ij}>0$ for good $j$, set
 $w_{ij}=1,r_{i'j}=1$ for two (and only those two) arbitrary buyers $i\ne i'$
 among the winners. We now show that all equations are satisfied.
 Constraint~\eqref{eqn:budget} is implied by the third condition of pacing
 equilibria. Constraint~\eqref{eqn:ydef} holds since we set $y_i=1$ exactly when
 buyer $i$ spends the whole budget. Constraint~\eqref{eqn:alpha} is implied by
 our choice of $y_i$ combined with the third condition of pacing equilibria.
 Constraint~\eqref{eqn:spend} is implied by the first condition of pacing
 equilibria. Constraint~\eqref{eqn:xdef} is implied by the third condition of
 pacing equilibria combined with the fact that buyers spend nothing on a good
 unless they are allocated a non-zero amount. Constraint~\eqref{eqn:highbid}
 and~\eqref{eqn:highbidupper} are implied by our choice of $h_j$ being the
 highest bid on good $j$ and the fact that $\bar{v}_j$ upper-bounds $v_{ij}$.
 Constraint~\eqref{eqn:winner} is implied by our choice for $w_{ij},x_{ij}$.
 Constraint~\eqref{eqn:price} is satisfied because we set $p_j$ equal to the second
 price, and the constraint is disabled for the highest bid due to $w_{ij}=1$
 and $v_{ij}$ being an upper bound on $\alpha_iv_{ij}$.
 Constraint~\eqref{eqn:priceupper} is implied by our choice of setting $r_{ij}=1$
 only if buyer $i$ constitutes the second price, and the fact that the
 constraint is disabled for all other buyers. Constraints~\eqref{eqn:onewinner},
 \eqref{eqn:onerunnerup}, and~\eqref{eqn:distinctwr} are implied by our choices for
 $w_{ij},r_{ij}$, respectively.

 Now assume that we have some satisfying assignment to the MIP. To construct a
 pacing equilibrium, assign pacing multipliers and spendings according to the
 values from the MIP, and set $x_{ij}=s_{ij}/p_j$. We now show that each of the
 three conditions for a pacing equilibrium are satisfied.

 Constraint~\eqref{eqn:spend} implies $\sum_{i\in N}x_{ij}=
 \sum_{i\in N}s_{ij}/p_j=1$. If $x_{ij}>0$ then $s_{ij}>0$ and by
 \eqref{eqn:xdef} $x_{ij}=1$, therefore \eqref{eqn:highbid}
 and~\eqref{eqn:highbidupper} imply $\alpha_iv_{ij}=h_j$. For all buyers $i'$
 with $x_{i'j} = 0$, we have $\alpha_{i'}v_{i'j}\leq \alpha_iv_{ij}$, otherwise
 we would violate \eqref{eqn:highbid} and thereby contradict our
 assumption of having a satisfying assignment. This shows that the first
 condition of a pacing equilibrium is satisfied.

 We first show that, in a
 feasible assignment $p_j$ must be equal to the second price. $p_j$ is both
 upper and lower-bounded by $\alpha_iv_{ij}$ for the buyer $i$ such that
 $r_{ij}=1$. Furthermore, \eqref{eqn:price} guarantees that $p_j$ is at
 least as high as the second-highest bid. Finally note that if $\alpha_iv_{ij}$
 is the highest bid $h_j$ and $r_{ij}=1$, then there must exist at least one
 other buyer such that $\alpha_{i'}v_{i'j}=h_j$ because
 \eqref{eqn:distinctwr} ensures that $w_{i'j}=1$ for some $i'$, and
 \eqref{eqn:highbidupper}-\eqref{eqn:winner} then imply that buyer $i'$
 must satisfy $\alpha_{i'}v_{i'j}=h_j$. This shows that $p_j$ is the second
 price. Now it remains to note that all buyers $i$ with $x_{ij}>0$ pay $p_j$,
 which is exactly the highest bid other than their own for $r_{ij}=0$. When
 $r_{ij}=1$, we established that $w_{i'j}=1$ for some other buyer, and thus
 $i$ and $i'$ must be tied for first price, and buyer $i$ is thus still paying
 the highest bid other than their own. This shows that the second condition of
 a pacing equilibrium is satisfied.

 Constraint~\eqref{eqn:budget} ensures that all budgets are satisfied.
 Constraints~\eqref{eqn:ydef} and~\eqref{eqn:alpha} ensure that if budgets are not
 fully spent then $y_i=0$, and $\alpha_i$ is then forced to be $1$. This shows
 that the third condition of a pacing equilibrium is satisfied.
\hfill\Halmos
\end{proof}

\section{Multiple Equilibria with Stochastic Valuations}\label{sec:stoch-valuations}

We provide two examples of pacing games with multiple equilibria where valuations are stochastic.

\subsection{Stochastic Pacing Games Equivalence}

In this section we show that our deterministic model has a natural stochastic analogue where valuations are not known ahead of time. In this model, we assume that there are $n$ buyers as before. We also assume that there are $m$ categories of goods, and for each category $j$, buyer $i$ has a weight $w_{ij} \geq 0$ (the weight vector $w_i$ is analogous to the valuation vector $v_i$ of a buyer in a deterministic instance). The buyers participate in $\ell$ auctions. For each auction, a category $j$ is sampled uniformly at random from the $m$ categories, and a \emph{quality score} $q \in [0, \omega_j]$ is  sampled from a distribution with mean $1$. The value that buyer $i$ has for the good is then $q\cdot w_{ij}$, where $j$ is the sampled category. Each buyer has some budget $B_i$. We want to find a vector of pacing multipliers $\alpha$ and a tie-breaking rule $x: [m] \rightarrow [0,1]^n$ specifying how each category is split in case of a tie, such that each buyer's budget constraint holds in expectation. Formally: 
\[
  \ell \cdot \mathbb{E}[x_i(j) p_j] \leq  B_i
\]

The above stochastic model turns out to be equivalent to a deterministic pacing game where each buyer has a valuation vector $w_i$, and a budget of $\frac{m}{\ell}B_i$. In particular, pacing equilibria in the deterministic pacing game and in the stochastic model are in a one-to-one correspondence:
\begin{proposition}
  Any vector of pacing multipliers $\alpha$ and corresponding tie-breaking function $x(j)$ is a pacing equilibrium of a stochastic pacing game if and only if it is a pacing equilibrium in the corresponding deterministic game.
\end{proposition}
\begin{proof}{Proof.}
  To show the proposition, it is sufficient to prove that 
  \[
    \mathbb{E}\left[ x_i(j) p_j \right]
    = \frac{\ell}{m} \sum_{j} x_{ij} p_j,
  \]
  since this implies that the pacing complementarity condition holds in the stochastic model if and only if it holds in the deterministic model.
To that end, let $k(j)$ be the index of the buyer $k$ with the second-highest bid $\alpha_kqw_{kj}$ for category $j$. We have
\begin{align*}
    \ell\cdot \mathbb{E}_{j}\left[ x_i(j) p_j \right]
  = \ell \cdot \frac{1}{m}\sum_j x_{ij} \alpha_{k(j)} w_{k(j)j} \mathbb{E}[q]
  =  \frac{\ell}{m} \sum_j x_{ij} \alpha_{k(j)} w_{k(j)j} 
  = \frac{\ell}{m} \sum_j  x_{ij} p_j.
\end{align*}
\hfill\Halmos
\end{proof}

Next, let us see an example of how this reduction can be applied to one of our examples from the paper. We consider a subset of the valuations used in Example~\ref{ex:revenue_equilibria}
shown to the left of Table~\ref{tab:correlated_cont_cdf}. Remember that in this example we have two very different pacing equilibria: $\alpha=(1,0.01)$ or $\alpha'=(0.01, 1)$.

\begin{table}[h]
  \caption{Valuations in the infinite goods example. Left: Deterministic valuations. Right: Stochastic valuations.}
  \label{tab:correlated_cont_cdf}
  \centering
  \begin{tabular}{c|c|c|c}
    & $g_1$ & $g_2$ & $g_3$ \\
    \hline
    $v_1$ & 100 & 1 & 99\\
    \hline
    $v_2$ & 1 & 100 & 99\\
    \hline
  \end{tabular}
  $\qquad$
  \begin{tabular}{c|c|c|c}
    & $g_1=1/3$ & $g_2=1/3$ & $g_3=1/3$ \\
    \hline
    $v_1$ & Unif$(0,200)$& Unif$(0,2)$ & Unif$(0,198)$\\
    \hline
    $v_2$ & $\frac{v_1}{100}$ & $100\cdot v_1$ & $v_1$\\
    \hline
  \end{tabular}
\end{table}

We can convert this discrete instance into a stochastic model by sampling valuations for $v_1$ and $v_2$ in correlated fashion. In particular, applying our reduction from above and sampling the quality score for each category from Unif$(0,2)$, we get a stochastic setting where we sample columns uniformly as indicated to the right of Table~\ref{tab:correlated_cont_cdf}.
Here, $v_1$ is chosen by uniformly sampling one of three uniform distributions,
and then sampling from the chosen distribution. Each uniform distribution
corresponds to one of the goods in the deterministic values in
Table~\ref{tab:correlated_cont_cdf}, and the expected value of each uniform
distribution is equal to the corresponding deterministic value. 
The valuation $v_2$ is then a deterministic function of $v_1$, chosen to preserve the relative relationship
between values from Table~\ref{tab:correlated_cont_cdf}. The marginal CDF of
$v_2$ is the same as that of $v_1$. We sample $n$ goods and set budgets equal to
$\frac{n}{3}$. In the stochastic setting we wish to satisfy the budget
constraint in expectation. This is achieved by both of the solutions from
before, $\alpha=(1,0.01)$ or $\alpha'=(0.01, 1)$, by linearity of expectation.

\subsection{Independently Drawn Valuations}

In the second example, we consider the following instance with independently drawn valuations. There are two buyers with budgets $B_1=1.2$ and $B_2=1$, respectively, and a single good.  
Each buyer samples their valuation uniformly iid from:
\[
  v_i \sim
  \begin{cases}
    40 & \frac{1}{2} \\
    4 & \frac{1}{2}
  \end{cases}.
\]
We show two very different pacing solutions that are equilibria. In the first one, both buyers have the same multipliers.
  Setting $\alpha_1=\alpha_2=\alpha$, the total spend is
$
  \alpha \frac{1}{4}(40 + 4 + 2\cdot 4)
  =
  13 \alpha
$.
It follows that if we set $\alpha = \frac{2.2}{13}$ then exactly $B_1+B_2=2.2$ is spent. A buyer wins when their value is 40 and the other one has value 4, and they may split the good when tied at $(40,40)$ and $(4,4)$. In order for budgets to be spent exactly, we solve for the fraction $x$ of the tied goods that buyer 2 should receive as follows:
\[
  \frac{2.2}{13} (1 + 10x + x) = 1,
\] 
which yields $x = \frac{54}{121}$.

In the second equilibrium, buyers have different multipliers.
We set $\alpha=(1, 0.1)$, in which case buyer 1 wins for every valuation vector except $(4,40)$; we allocate all of the good to buyer 2 in the case of the valuation vector $(4, 40)$. The spend of buyer $1$ is then
$
  \frac{0.1}{4} (40 + 4 + 4)
  =
  1.2
$,
while the spend of buyer $2$ is $\frac{1}{4}4=1$. 

Buyer 2 achieves expected utility of $\frac{54}{121}(10+1) + 10 - 1 \approx 14.9 - 1 \approx 13.9$ in the first equilibrium, but an expected utility of $9$ in the second.

The above example was for the case of a single good. However, we can extend this example to multiple goods as follows: say that there are $m$ goods in total, and each buyer draws their valuation $v_{ij}$ iid according to the same distribution as before; i.e.,
\[
  v_{ij} \sim
  \begin{cases}
    40 & \frac{1}{2} \\
    4 & \frac{1}{2}
  \end{cases}.
\]

Now, we set the budgets equal to $B_1 = 1.2m$ and $B_2=m$, respectively. By linearity of expectation, the equilibria from the single-good case still cause each buyer to exactly spend their budget in expectation, and thus are equilibria for the $m$-good case as well.

\section{Iterating Best Responses Cycles}\label{sec:cycling}

Here, we provide step-by-step details of Example~\ref{ex:bids_appendix} which shows that iterating best responses may cycle.
\begin{itemize}
	\item Initially, buyer 1 wins auctions 1 and 5 and pays 60; buyer 2 wins auctions 2, 3, 4, and 6 and pays 1928. buyer 2 exceeds its budget of 1300 at these multipliers---it exhausts its budget from auction 2 alone, in which it pays 1300, and it also wins three other auctions. Buyer 2's best response is to lower its multiplier so that it wins only auction 2. To do so, buyer 2 sets its multiplier somewhere on the interval $(1300/6503, 5/25) \approx (0.1999, 0.2)$: any lower, and its bid for auction 2 drops below buyer 1's bid of 1300, in which case buyer 2 wins nothing; any higher, and its bid for auction 6 exceeds buyer 3's bid of 5, in which case buyer 2 exceeds its budget.
	\item After buyer 2 lowers its multiplier, buyer 1 wins more auctions: In addition to what it was winning previously, buyer 1 also wins auction 3 at a price equal to buyer 2's paced bid of at least $300.6(1300/6503) \approx 60.09$. Buyer 1 exhausts its budget of 60 from auction 3 alone. Buyer 1 must set its multiplier low enough to not win auction 3, but such a multiplier is so low that it results in buyer 1 losing all other auctions. Buyer 1's best response is to tie on auction 3, where buyer 2's paced bid is at most $300.6(5/25)$. To do so, buyer 1 sets its multiplier to at most $300.6(5/25) / 123 \approx 0.488$.
	\item After buyer 1 lowers its multiplier, buyer 2 goes from losing to tying on auction 3, causing buyer 2 to pay more than it was previously for that auction, but it also pays much less for auction 2: Instead of paying 1300 for auction 2 as it was previously, it pays around $1300(0.488) = 634.4$. Because buyer 2 is paying so much less for auction 2, it can raise its multiplier to 1, causing it to win auctions 3, 4, and 6 and to pay less than its budget.
	\item After buyer 2 raises its multiplier to 1, buyer 1 no longer wins auction 3. It can raise its multiplier to 1 and still not exhaust its budget. This brings us back to the first iteration, where all multipliers were set to 1.
\end{itemize}

\section{Experiments on Best-Response Dynamics (One-Shot Setting)}\label{sec:brdynamics}

We considered \emph{best-response dynamics} (BR dynamics) to search for a pacing equilibrium
in the standard, one-shot setting. We briefly describe these experiments here.
BR dynamics can be thought of as a repeated
auction market where each buyer has some budget to spend every day and wishes
to set its pacing multiplier appropriately. At the end of each day, buyers
observe the outcome for the day and best respond to the strategy of the other
players. Our goal in these experiments was to see whether warm-starting BR
dynamics with the MIP output can improve convergence of BR dynamics and lead it
to outcomes with higher welfare than it would otherwise achieve.

We consider two BR algorithms that differ in how the best response is
computed. If there is more than one BR pacing multiplier, we break ties towards
the highest pacing multiplier (\emph{BR high}\/), or towards the lowest
(\emph{BR low}\/). Both algorithms always start from the same random
initialization of pacing multipliers. In addition, we consider BR high starting
from the MIP solutions and refer to it as \emph{Init MIP}. When needed, we
replace MIP in the name by a specific MIP objective. For the BR setting, we
consider random tiebreaking rather than having fractional allocations be part of
the bids. Thus, a pacing equilibrium might not be stable if it includes
fractional allocations.
We evaluated the BR algorithms on a subset of 50 synthetic instances taken randomly from those in the computational study.

\begin{figure}[t!]
  \centering
	\includegraphics[width=0.4\columnwidth]{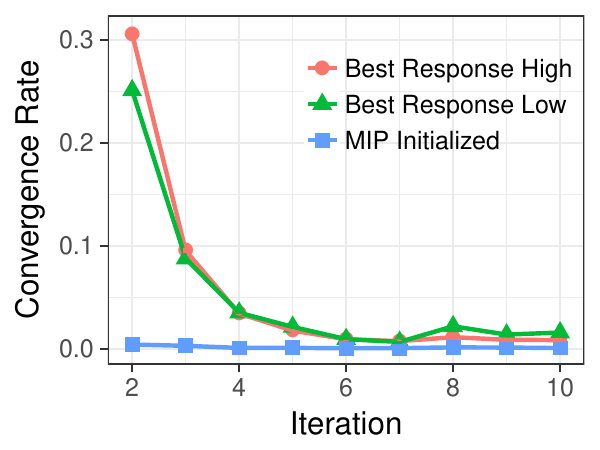}
	\caption{BR dynamics convergence rate. For each iteration, we show
    the absolute difference in a buyer's multipliers
    from the previous iteration, averaged across buyers and instances.
    For MIP initialization we average across solutions from
    all objectives.}
	\label{fig:brstep}
\end{figure}

\begin{figure}[t!]
  \centering
	\includegraphics[width=0.5\columnwidth]{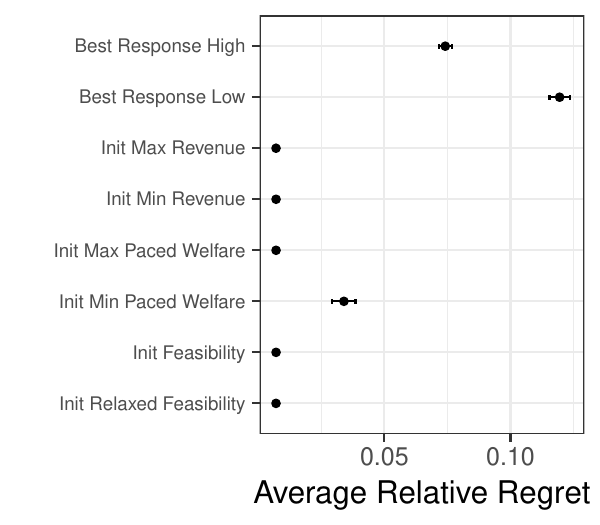}
  \caption{Maximum relative regret over buyers in an instance for various BR
    algorithms averaged across instances.}
	\label{fig:avgmaxrelregret}
\end{figure}

We start by looking at BR dynamics convergence and regret.
Figure~\ref{fig:brstep} shows that the BR algorithms converge quickly in our
computational study. They required less than $10$ iterations to reach small
oscillations in pacing multipliers. Figure~\ref{fig:avgmaxrelregret} shows the
maximum relative regret across all buyers, averaged across instances. The
relative regret for a buyer is computed as the ratio of the utility-improvement
they could get by best responding, divided by the utility of the best response
(i.e., the fraction of utility they are missing out on).
For the purposes of computing regret, when a buyer exceeds its budget, we do not set utility to negative infinity; instead we penalize utility by the amount over budget
multiplied by the spend-to-budget ratio times paced-welfare-to-budget ratio. We
see that both BR high and BR low have somewhat high relative regret, missing
out on $7.5\%$-$12\%$ utility. Contrary to this, Init MIP
solutions perform well and are able to stay near equilibrium for most
instances.

Figure~\ref{fig:maxrelregret} shows the relative regret broken down by each algorithm. This plot shows that the poor performance of initializing with the objective that minimized paced welfare was actually caused by a single outlier.

\begin{figure}[t!]
	\includegraphics[width=\columnwidth]{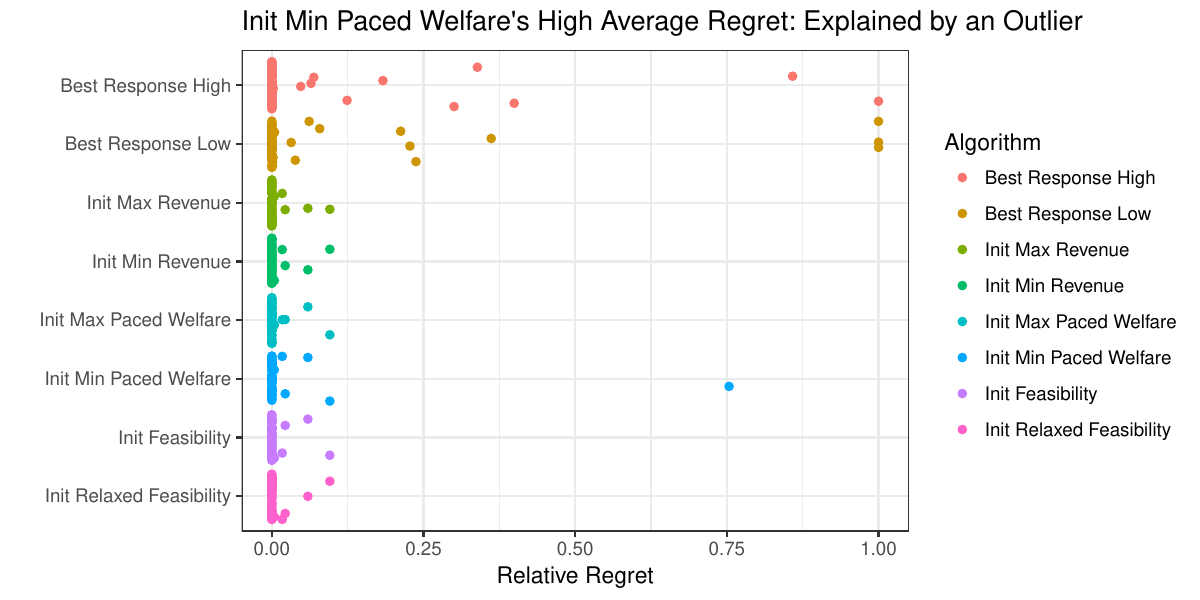}
	\caption{Relative regret broken down by each algorithm. Each point represents a BR algorithm running on a particular problem instance.}
	\label{fig:maxrelregret}
\end{figure}

Finally, we look at the improvement in market outcomes from seeding BR dynamics with the MIP output. Figure~\ref{fig:brperf_appexp} shows the revenue, welfare, and paced welfare achieved by the different BR dynamics algorithms relative to the MIP. Each point in the plot shows the average performance of a given algorithm relative to the solution maximizing each objective. BR low
performs significantly worse than BR high across all three dimensions. For
revenue and welfare, they both perform significantly worse than the MIP
solutions as well, in spite of the fact that the BR solutions may not even
respect budgets. The BR dynamics perform significantly better with MIP initializations than without.

\begin{figure}[t!]
  \centering
	\includegraphics[width=0.7\columnwidth]{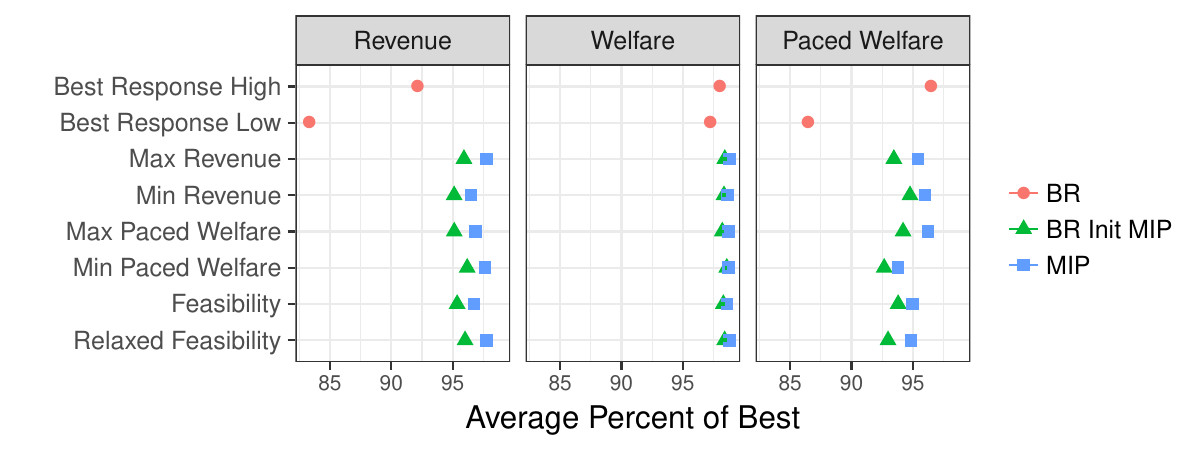}
	\caption{Performance of BR and MIP algorithms across 50 instances. The x-axis
    shows the percentage of the best value for each objective, averaged across
    buyers and instances.}
	\label{fig:brperf_appexp}
\end{figure}

\section{Additional Experimental Results and Details}\label{sec:experiment:appendix}

This section describes our experimental setup in more detail than space permitted in the main body of the paper and provides additional results.

\subsection{Equilibrium Gaps}

The experiments section reported on maximum gaps for different objectives and instance distributions. Here, we show additional summary statistics. Figure~\ref{fig:irperfgroup} shows the relative gap compared to optimal solutions for equilibria maximizing or minimizing the different objectives, measured with respect to each objective, grouped by instance type.

\begin{figure}
	\includegraphics[width=\columnwidth]{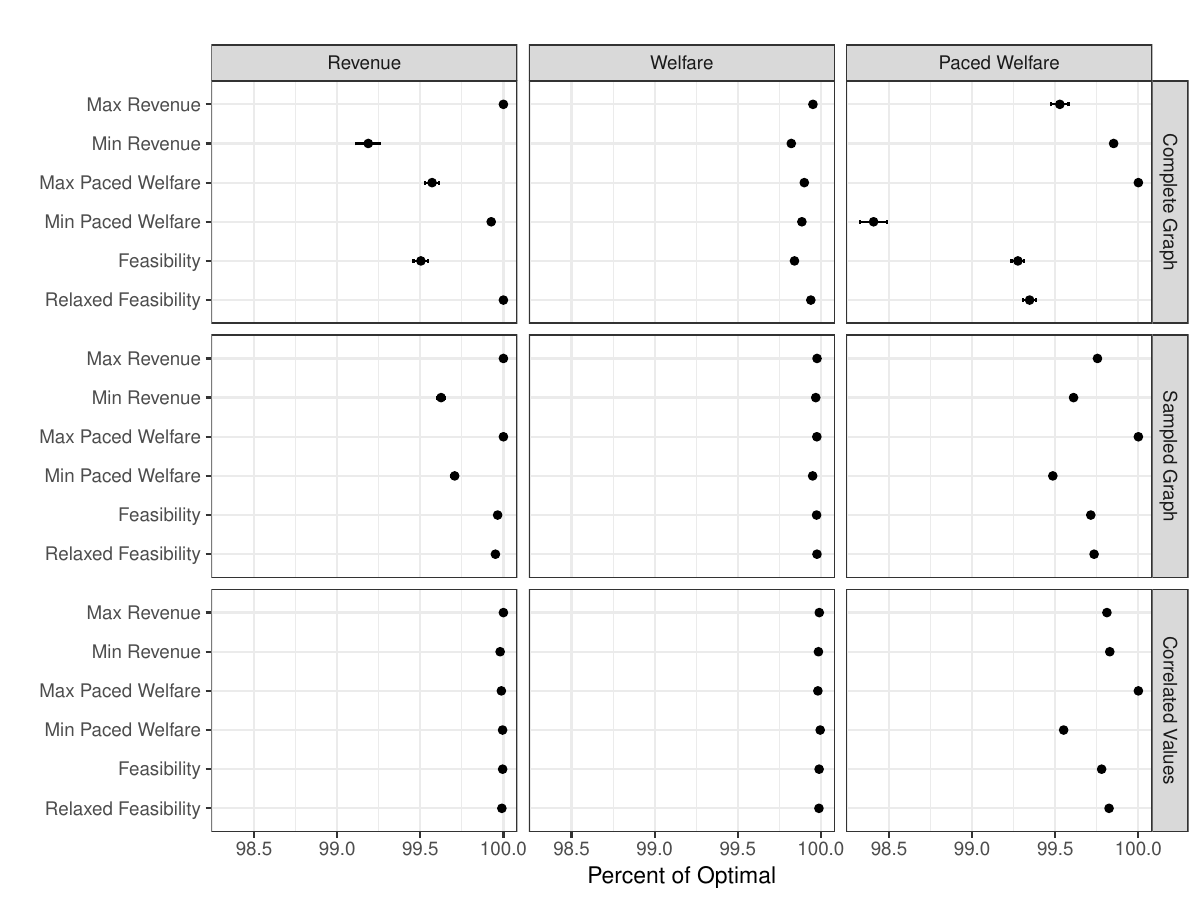}
	\caption{Percentage optimality for equilibria maximizing or minimizing the different objectives, measured with respect to each objective, grouped by instance type.}
	\label{fig:irperfgroup}
\end{figure}

\subsection{Other Terms and Notation}

We informally defined some terms in the experiments section, which we now define more precisely. For a given scaled-up instance $\tilde{\Gamma}$, let $k_j \in M$ be the \emph{good type} of auction $j$; this good type associates the auction in the scaled-up instance with a good in the original instance.
Let $\tilde{M}_j \subseteq \tilde{M}$ be the set of auctions in the scaled-up instance that have good type $j$ (i.e., $\tilde{M}_j = \{j' \in \tilde{M} : k_{j'} = j\}$ for $j \in M$). For a given run of \AdaptiveAlg on a scaled-up instance, let $x_{ij}'$ be the \emph{empirical allocation} over good types: the fraction of auctions that buyer $i$ won for good type $j \in M$. That is, let $x_{ij}' = \left( \sum_{j' \in \tilde{M}_j} \tilde{x}_{ij'} \right) / \left( \sum_{i' \in \tilde{N},j' \in \tilde{M}_j} \tilde{x}_{ij'} \right)$, where $\tilde{x}_{ij'}$ is \AdaptiveAlg's output allocation ($\forall i \in \tilde{N}, j' \in \tilde{M}$).
For a given run of \AdaptiveAlg on a scaled-up instance, let a buyer's \emph{regret} be the difference between the buyer's maximum possible utility in hindsight (given fixed bids of other buyers) and the buyer's realized utility; let the \emph{max regret} be the maximum such regret across all buyers.

\section{Pacing Dynamics}
\label{se:dynamics}
While there are no dynamics in the definition of our game, we consider dynamics to evaluate the quality of the solutions
provided by the equilibrium concept.
It is instructive to consider the definition of pacing equilibrium in the
context of dynamics.  Specifically, suppose that the goods are sold
continuously over the period $[0,1]$.  I.e., at time $t \in [0,1]$ a
fraction $t$ of every good will have been sold.  Within each
infinitesimal slice of time a second price auction is used for each
infinitesimal fraction of a good; if there is a tie for a good then
it may be split into arbitrary fractions $x_{ijt}$ among the buyers,
summing to $1$ if there are positive bids.  In an ad auction, this
would correspond to the limit case where there are large numbers of
all types of impressions, and the distribution of such types does not
vary over time.  Then, we can consider $\alpha_i$ to change
dynamically over time (so we get $\alpha_{it}$).  Specifically, if a
buyer $i$ is currently spending at a rate that will overspend her
remaining budget over the remaining period $[t,1]$, we decrease
$\alpha_{it}$; if it will underspend {\em and} $\alpha_{it}<1$, then
we increase $\alpha_{it}$.  Call this the {\em limit dynamics model}.
\begin{definition}
Multipliers $\alpha_i \in [0,1]$ and fractions $x_{ij} \in [0,1]$
constitute a {\em stable solution} in the limit dynamics model if
setting $\alpha_{it}=\alpha_i$ and $x_{ijt}=x_{ij}$ (for all $i,j,t$)
satisfies the feasibility conditions for the $x_{ijt}$ and is
consistent with the dynamics (i.e., no $\alpha_{it}$ ever needs to be
adjusted up or down).
\end{definition}
\begin{proposition}
Multipliers $\alpha_i$ and fractions $x_{ij}$ constitute a stable
solution in the limit dynamics model if and only if they constitute a
pacing equilibrium.
\end{proposition}
\begin{proof}{Proof.}
Suppose they constitute a pacing equilibrium.  Then, $x_{ijt}$ is
nonzero only if $\alpha_{ijt}v_{ij}=\alpha_{ij}v_{ij}$ is one of the
highest bids, and for any $j,t$, we have $\sum_{i} x_{ijt} = \sum_{i}
x_{ij} \leq 1$ with equality if there is at least one positive bid.
For a buyer with $\sum_j s_{ij}=B_i$ in the pacing equilibrium, we
also have $\int\limits_{t=0}^1 \sum_j s_{ijt} = 1 \cdot \sum_j s_{ij}
= B_i$, so the buyer is always exactly on track to spend her budget
and the multiplier need not be adjusted.  For a buyer with $\sum_j
s_{ij}<B_i$ in the pacing equilibrium we must have $\alpha_i=1$; we
have $\int\limits_{t=0}^1 \sum_j s_{ijt} = 1 \cdot \sum_j s_{ij} <
B_i$, so the buyer is always on track to underspend (which is fine
because $\alpha_{it}=\alpha_i=1$).  Hence they constitute a stable
solution.
Conversely, suppose they constitute a stable solution.  Then $\sum_{i}
x_{ij} = \sum_{i} x_{ij0} \leq 1$ with equality if there is at least
one positive bid. We also have $p_{j}=s_{ij}/x_{ij}
=s_{ij0}/x_{ij0}=p_{j0}$ which is the second-highest bid $\alpha_{i0}
v_{ij} = \alpha_{i} v_{ij}$.  For any buyer $i$, $\sum_j s_{ij} =
\int\limits_{t=0}^1 \sum_j s_{ijt} \leq B_i$.  Finally, if
$\alpha_i<1$ then $\sum_j s_{ij} = \int\limits_{t=0}^1 \sum_j s_{ijt}
= B_i$ (otherwise the multiplier would be adjusted and we would not
have $\alpha_{ijt}=\alpha_{ij}$ for all $t$).
\hfill\Halmos
\end{proof}

\iffalse
\section{Proof of Theorem \ref{marg-butt-th}.}

To avoid confusion, theorems that we repeat for readers' convenience
will have the same appearance as when they were mentioned for the first
time. However, here they should be coded by \verb+repeattheorem+ instead
of \verb+theorem+ to keep labels/pointers uniquely resolvable. Other predefined theorem-like
environments work similarly if they need to be repeated in what becomes the \mbox{e-companion.}

\proof{Proof of Theorem 1.} 
The first statement is a conseqence of Lemmas~\ref{aux-lem1} and \ref{aux-lem2}. The rest relies on the fact that the continuous image of a compact set into the reals is a closed interval, thus having a minimum point.
\Halmos \endproof

\begin{proof}{Proof.}
\hfill\Halmos
\end{proof}

\fi

% Appendix here
% Options are (1) APPENDIX (with or without general title) or
%             (2) APPENDICES (if it has more than one unrelated sections)
% Outcomment the appropriate case if necessary
%
% \begin{APPENDIX}{<Title of the Appendix>}
% \end{APPENDIX}
%
%   or
%
% \begin{APPENDICES}
% \section{<Title of Section A>}
% \section{<Title of Section B>}
% etc
% \end{APPENDICES}

%%%%%%%%%%%%%%%%%
\end{document}